\documentclass[a4paper,USenglish, cleveref, autoref, thm-restate, nolineno]{socg-lipics-v2021}

\usepackage{tcolorbox} 
\usepackage{algorithm}
\usepackage{algorithmicx}
\usepackage{algpseudocode}
\usepackage{multirow}
\usepackage{tikz}
\usetikzlibrary{calc}

\usepackage{hyperref}
\hypersetup{
    colorlinks=true,
    linkcolor=red,
    filecolor=magenta,
    citecolor = blue,
    urlcolor=cyan,
    linktocpage,
    }

\algtext*{EndWhile}% Remove "end while" text
\algtext*{EndIf}% Remove "end if" text
\algtext*{EndFor}% Remove "end for" text

\title{Approximating Convex Hulls via Range Queries}

\author{Thomas Schibler}{University of California, Santa Barbara, USA}{tschibler@ucsb.edu}{0009-0008-2966-9468}{}
\author{Jie Xue}{New York University Shanghai, China}{jiexue@nyu.edu}{0000-0001-7015-1988}{}
\author{Jiumu Zhu}{New York University Shanghai, China}{jz5348@nyu.edu}{0000-0001-9554-5908}{}

\authorrunning{T. Schibler, J. Xue, J. Zhu}

\usepackage{xifthen}

\hideLIPIcs

\begin{document}

%\pagenumbering{gobble}

\maketitle

\begin{abstract}
Recently, motivated by the rapid increase of the data size in various applications, Monemizadeh [APPROX'23] and Driemel, Monemizadeh, Oh, Staals, and Woodruff [SoCG'25] studied geometric problems in the setting where the only access to the input point set is via querying a range-search oracle.
Algorithms in this setting are evaluated on two criteria: (i) the number of queries to the oracle and (ii) the error of the output.
In this paper, we continue this line of research and investigate one of the most fundamental geometric problems in the oracle setting, i.e., the convex hull problem.

Let $P$ be an unknown set of points in $[0,1]^d$ equipped with a range-emptiness oracle.
Via querying the oracle, the algorithm is supposed to output a convex polygon $C \subseteq [0,1]^d$ as an estimation of the convex hull $\mathcal{CH}(P)$ of $P$.
The error of the output is defined as the volume of the symmetric difference $C \oplus \mathcal{CH}(P) = (C \backslash \mathcal{CH}(P)) \cup (\mathcal{CH}(P) \backslash C)$.
We prove tight and near-tight tradeoffs between the number of queries and the error of the output for different variants of the problem, depending on the type of the range-emptiness queries and whether the queries are non-adaptive or adaptive.
\smallskip
\begin{itemize}
    \item \textbf{Orthogonal emptiness queries in $d$-dimensional space:} \\
    We show that the minimum error a deterministic algorithm can achieve with $q$ queries is $\Theta(q^{-1/d})$ if the queries are non-adaptive, and $\Theta(q^{-1/(d-1)})$ if the queries are adaptive.
    In particular, in 2D, the bounds are $\Theta(1/\sqrt{q})$ and $\Theta(1/q)$ for non-adaptive and adaptive queries, respectively.
    \smallskip
    
    \item \textbf{Halfplane emptiness queries in 2D:} \\
    We show that the minimum error a deterministic algorithm can achieve with $q$ queries is $\Theta(1/\sqrt{q})$ if the queries are non-adaptive, and $\widetilde{\Theta}(1/q^2)$ if the queries are adaptive.
    Here $\widetilde{\Theta}(\cdot)$ hides logarithmic factors.
\end{itemize}
\end{abstract}

%\pagenumbering{arabic}

\section{Introduction}

Classic algorithms are designed to compute solutions to a problem instance by processing the \textit{entire} input dataset.
These algorithms can suffer two potential drawbacks.
The first drawback concerns \textit{efficiency}.
As all of the input data has to be received and examined, the time complexity of such algorithms (even the most efficient ones) is at least linear in the input size.
However, due to the rapidly growing size of the datasets involved in real-world applications nowadays, linear running time is already not satisfactory in many scenarios.
The second drawback regards \textit{data privacy}.
For security reasons, sometimes the users would like to keep their data private and therefore cannot directly provide the exact dataset to the algorithm.
Instead, they can only provide partial and implicit information of the dataset and ask the algorithm to give useful results based on the information provided.
In this situation, the classic algorithms no longer work due to lack of full information of the dataset.

Motivated by removing these drawbacks, researchers have studied algorithms in the \textit{oracle model}.
In this model, the algorithm does not have direct access to the input dataset. Instead, a provided oracle built on the dataset can answer certain queries about the dataset.
The algorithm is required to compute a good solution by performing queries to the oracle.
Algorithms in this setting are usually evaluated on two criteria: (i) the number of queries to the oracle and (ii) the error of the output.

For geometric problems, various types of oracles have been considered in the literature~\cite{chazelle2003sublinear,czumaj2001property,czumaj2005approximating,driemel2025range,monemizadeh2023facility}.
When the input dataset is a set $P$ of points in a Euclidean space $\mathbb{R}^d$, a natural type is the \textit{range-search} oracle.
A query to a range-search oracle is a \textit{range} $Q$ in the space $\mathbb{R}^d$ of a specific shape, and the oracle will return certain information about the points in $P \cap Q$.
For example, a \textit{range-counting} oracle returns $|P \cap Q|$~\cite{afshani2007approximate,chan2011orthogonal}, a \textit{range-emptiness} oracle returns whether $P \cap Q = \emptyset$ or not~\cite{erickson1997space,sharir2011semialgebraic}, a \textit{range-reporting} oracle returns the set $P \cap Q$ itself~\cite{agarwal1995dynamic,ramos1999range}, etc.
Range-search oracles have the following advantages.
First, range search can usually be implemented very efficiently.
As a fundamental topic in Computational Geometry, range search has been extensively studied over decades and many efficient data structures have been proposed for various range queries~\cite{agarwal1999geometric,agarwal1994range,bentley1980efficient,chan2011orthogonal,chazelle1990lower,matouvsek1994geometric}.
Furthermore, most types of range search (except range reporting) do not reveal the exact data points inside the query range $Q$, and therefore data privacy is well guaranteed.
%are designed to process the entire input to get exact or approximate
%solutions, the increasing demand for efficiency in massive datasets necessitates approaches that can provide useful results without examining all of the input data.
%Motivated by this, a wide range of sublinear-time algorithms has been extensively studied over the past
Recently, Monemizadeh~\cite{monemizadeh2023facility} and the authors of \cite{driemel2025range} studied multiple geometric problems with range-search oracles, following earlier work of Czumaj and Sohler~\cite{czumaj2001property} and the authors of \cite{czumaj2005approximating}.
Problems considered include facility location~\cite{monemizadeh2023facility}, clustering~\cite{czumaj2001property}, Euclidean minimum spanning tree~\cite{czumaj2005approximating,driemel2025range}, earth mover distance~\cite{driemel2025range}, etc.
They show that for all these problems, one can obtain nontrivial approximation solutions via a small number of range-search queries.

In this paper, we continue this line of research and investigate one of the most fundamental geometric problems, the \textit{convex hull} problem~\cite{avis1995good,barber1996quickhull,brown1979voronoi,chan1996optimal,seidel1981convex,yao1981lower}, in the range-search oracle model.
In the convex hull problem, the goal is to compute the \textit{convex hull} of a set $P$ of points in $\mathbb{R}^d$, denoted by $\mathcal{CH}(P)$, which by definition is the smallest convex body in $\mathbb{R}^d$ containing $P$.
We study the problem with the simplest type of range-search oracles, i.e., \textit{range-emptiness} oracles.
We are interested in finding a convex body $C^*$ as an approximation of $\mathcal{CH}(P)$ via a small number of range-emptiness queries.
To this end, we need a measure for the error of the approximation.
The most natural measure one can use is the relative error $\lVert C^* \oplus \mathcal{CH}(P) \rVert\ /\ \lVert \mathcal{CH}(P) \rVert$, where $\lVert \cdot \rVert$ denotes the volume and $C^* \oplus \mathcal{CH}(P) = (C^* \backslash \mathcal{CH}(P)) \cup (\mathcal{CH}(P) \backslash C^*)$ is the symmetric difference between $C^*$ and $\mathcal{CH}(P)$.
Unfortunately, one can easily see that it is impossible to approximate the convex hull with any bounded relative error no matter how many range-emptiness queries the algorithm performs\footnote{Even in $\mathbb{R}^1$, there is no way to check whether $P$ only contains one point, in which case $\lVert \mathcal{CH}(P) \rVert = 0$, or $P$ contains at least two points, in which case $\lVert \mathcal{CH}(P) \rVert > 0$, via range-emptiness queries.}.
Therefore, we shall instead consider the additive error $\lVert C^* \oplus \mathcal{CH}(P) \rVert$.
Clearly, the additive error depends on the extent measure of $P$.
So we need an extra normalization assumption: we require all points in $P$ to lie in the unit hypercube $[0,1]^d$.
Below we formulate the problem to be studied.

\begin{tcolorbox}[title = {\textsc{Approximate Convex Hull} in $\mathbb{R}^d$ via range queries}]
\textbf{Input:} A (black-box) range-emptiness oracle $\mathcal{O}$ on a set $P$ of points in $[0,1]^d$ \\
\textbf{Output:} A convex body $C^*$ that approximates $\mathcal{CH}(P)$
\end{tcolorbox}

Regarding the above problem, a natural question concerns the tradeoff between the number of queries and the (additive) error of the output: if the algorithm is allowed to perform $q$ queries to the oracle $\mathcal{O}$, what is the minimum error it can achieve (in the worst case)?
The answer to this question depends on the following two features of the queries.
\smallskip
\begin{itemize}
    \item \textbf{Shape of the query ranges.}
    %The shape of query ranges matters a lot in range-search problems.
    Range queries with different shapes might behave very differently.
    The most commonly used queries are orthogonal queries~\cite{chan2011orthogonal}, where the query ranges are axis-parallel rectangles/boxes.
    Besides these, well-studied range queries include halfplane/halfspace queries~\cite{afshani2009optimal,edelsbrunner1986halfplanar}, simplex queries~\cite{haussler1986epsilon}, semi-algebraic queries~\cite{agarwal1994range,sharir2011semialgebraic}, etc.
    \smallskip

    \item \textbf{Adaptivity of the queries.}
    In the non-adaptive query model, the algorithm must make all its queries at once, and then make its estimation based on the batched answer.
    In the adaptive query model, the algorithm is allowed to make queries at any time and the oracle will provide the answer immediately.
    In particular, the next query can be made after seeing the answers of the previous queries.
    %a query, see the result, and then choose its next query accordingly, thus giving the algorithm the power to tailor its queries to each particular input.
\end{itemize}

As the main contribution of this paper, we prove tight and near-tight tradeoffs between the number of queries and the error of the output for \textsc{Approximate Convex Hull} with orthogonal emptiness queries in $\mathbb{R}^d$ for any fixed $d$ and halfplane emptiness queries in $\mathbb{R}^2$, in both non-adaptive and adaptive query models.
\smallskip
\begin{itemize}
    \item \textbf{Orthogonal queries.}
    We show the minimum error a deterministic algorithm can achieve for \textsc{Approximate Convex Hull} in $\mathbb{R}^d$ with $q$ orthogonal emptiness queries is $\Theta(q^{-1/d})$ if the queries are non-adaptive and is $\Theta(q^{-1/(d-1)})$ if the queries are adaptive.
    In particular, in 2D, the bounds are $\Theta(1/\sqrt{q})$ and $\Theta(1/q)$ for non-adaptive and adaptive queries, respectively.
    \smallskip
    
    \item \textbf{Halfplane queries.}
    We show the minimum error a deterministic algorithm can achieve for \textsc{Approximate Convex Hull} in $\mathbb{R}^2$ with $q$ halfplane emptiness queries is $\Theta(1/\sqrt{q})$ if the queries are non-adaptive and is $\widetilde{\Theta}(1/q^2)$ if the queries are adaptive\footnote{The notation $\widetilde{\Theta}(\cdot)$ hides factors logarithmic in $q$.}.
    \smallskip
\end{itemize}

\noindent
Table~\ref{tab-result} summarizes the tradeoffs we prove in this paper. All of our algorithms are deterministic and have offline time complexity (outside the oracle) polynomial in $q$. Our lower bounds similarly hold for deterministic algorithms.

\begin{table}[h]
    \centering
    \renewcommand\arraystretch{1.2}
    \begin{tabular}{c|c|c|c|c|c}
        \hline
        Query shape & Space & Query type & Upper bound & Lower bound & Source\\
        \hline
        \multirow{2}{*}{Orthogonal} &  \multirow{2}{*}{$\mathbb{R}^d$} & Non-adaptive & $O(q^{-1/d})$ & $\Omega(q^{-1/d})$ & Theorems~\ref{thm-nonadapt_orth_ub} and~\ref{thm-nonadapt_orth_lb} \\
        \cline{3-6}
        & & Adaptive & $O(q^{-1/(d-1)})$ & $\Omega(q^{-1/(d-1)})$ & Theorems~\ref{thm-adapt_orth_ub} and~\ref{thm-adapt_orth_lb} \\
        \hline
        \multirow{2}{*}{Halfplane} &  \multirow{2}{*}{$\mathbb{R}^2$} & Non-adaptive & $O(1/\sqrt{q})$ & $\Omega(1/\sqrt{q})$ & Theorems~\ref{thm-nonadapt_hplane_ub} and~\ref{thm-nonadapt_hplane_lb} \\
        \cline{3-6}
        & & Adaptive & $\widetilde{O}(1/q^2)$ & $\Omega(1/q^2)$ & Theorems~\ref{thm-adapt_hplane_ub} and~\ref{thm-adapt_hplane_lb} \\        
        \hline
    \end{tabular}
    \caption{Minimum error one can achieve for \textsc{Approximate Convex Hull} with $q$ queries.}
    \label{tab-result}
\end{table}

\vspace{-1cm}

\subparagraph{Related work.}
Both convex hulls and range search have been extensively studied in the literature.
See~\cite{agarwal2017range,seidel2017convex} for surveys of these topics.
%Providing an exhaustive summary of these topics is of course out of the scope of this paper.
Problems related to convex hulls have also been considered in oracle models prior to this paper.
For example, the celebrated work of Chazelle, Liu, and Magen~\cite{chazelle2003sublinear} considered the problem of approximating the volume of the convex hull in 2D and 3D via a sampling oracle, which can uniformly sample an input point.
Czumaj and Sohler~\cite{czumaj2001property} studied the problem of testing convex position via range queries.
Here the goal is to check whether one can remove an $\varepsilon$-fraction of the points from the input point set to make it in convex position.
The oracle used in~\cite{czumaj2001property} is slightly stronger than range-emptiness oracles --- it can report one point in the query range.

\section{Preliminaries}

\subparagraph{Basic notations.}
For a number $n \in \mathbb{N}$, we write $[n] = \{1,\dots,n\}$.
For a convex body $C$ in $\mathbb{R}^d$, we denote by $\lVert C \rVert$ its volume.
The notation $| \cdot |$ is used with different meanings depending on the context.
For a segment $\sigma$, we use $|\sigma|$ to denote the length of $\sigma$.
For an angle $\alpha$, $|\alpha| \in [0,2\pi)$ is the magnitude of $\alpha$.
Furthermore, for a vector $\vec{v}$, $|\vec{v}|$ denotes its magnitude.

\subparagraph{Vectors and halfspaces.}
%For two (different) unit vectors $\vec{u},\vec{v} \in \mathbb{S}^1$, let $\mathsf{ang}(\vec{u},\vec{v})$ denote the magnitude of the angle between $\vec{u}$ and $\vec{v}$ with counterclockwise boundary $\vec{u}$ and clockwise boundary $\vec{v}$.
For a halfspace $H \subset \mathbb{R}^d$, denote by $\partial H$ its bounding hyperplane.
(More generally, $\partial C$ denotes the boundary of any convex body $C$.)
The \textit{normal vector} of a halfspace $H$ is a unit vector $\vec{v} \in \mathbb{S}^{d-1}$ perpendicular to $\partial H$ so that for any $q \in \partial H$, $H = \{p \in \mathbb{R}^{d}: \langle p,\vec{v} \rangle \leq \langle q,\vec{v} \rangle \}$.

\section{Approximating convex hulls via orthogonal queries}

In this section, we present our results for \textsc{Approximate Convex Hull} in $\mathbb{R}^d$ via orthogonal emptiness queries.
For non-adaptive queries (Section~\ref{sec-nonadaptorth}), the algorithm is very simple, while the lower bound proof is nontrivial and interesting.
For adaptive queries (Section~\ref{sec-adaptorth}), the algorithm is more technical and the lower bound proof is simpler.

\subsection{Non-adaptive orthogonal queries} \label{sec-nonadaptorth}

We first present our algorithmic result with non-adaptive orthogonal emptiness queries.
Let $P \subseteq [0,1]^d$ be an unknown set of points and $\mathcal{O}$ be the orthogonal emptiness oracle on $P$.
For an axis-parallel box $\Box$ in $\mathbb{R}^d$, let $\textsc{Query}(\Box)$ denote the output of $\mathcal{O}$ when queried with $\Box$, which is $\mathsf{yes}$ if $P \cap \Box = \emptyset$ and is $\mathsf{no}$ if $P \cap \Box \neq \emptyset$.

\begin{algorithm}
    \caption{\textsc{NonAdaptiveOrthogonal}$(q)$}
    \begin{algorithmic}[1]
        \State $r \leftarrow \lfloor q^{1/d} \rfloor$
        \State $\varGamma \leftarrow \textsc{Partition}([0,1]^d,r)$
        \State $A_\Box \leftarrow \textsc{Query}(\Box)$ for all $\Box \in \varGamma$
        \State $\varGamma_1 = \{\Box \in \varGamma: A_\Box = \mathsf{no}\}$
        \State $C^* \leftarrow \mathcal{CH}(\bigcup_{\Box \in \varGamma_1} \Box)$
        \State \textbf{return} $C^*$
    \end{algorithmic}
    \label{alg-nonadabox}
\end{algorithm}

Our algorithm for approximating $\mathcal{CH}(P)$ is very simple (presented in Algorithm~\ref{alg-nonadabox}) and is similar to the algorithm of Czumaj and Sohler~\cite{czumaj2001property} for testing convex position.
Let $r = \lfloor q^{1/d} \rfloor$.
In line~2, the sub-routine $\textsc{Partition}([0,1]^d,r)$ partitions $[0,1]^d$ evenly into $r^d$ \textit{cells} each of which is a hypercube of side-length $\frac{1}{r}$; let $\varGamma$ be the set of the $r^d$ cells.
Then we query $\mathcal{O}$ with the boxes in $\varGamma$, and let $A_\Box = \textsc{Query}(\Box)$ for $\Box \in \varGamma$.
Define $\varGamma_1$ as the set of cells $\Box \in \varGamma$ with $A_\Box = \mathsf{no}$ (which are just the cells containing at least one point in $P$).
Finally, the algorithm simply returns $C^* = \mathcal{CH}(\bigcup_{\Box \in \varGamma_1} \Box)$.
Clearly, the number of queries to $\mathcal{O}$ is at most $q$, and they are non-adaptive as the algorithm performs them at the same time.

To bound the error of our algorithm, we need the following lemma about the volume of the Minkowski sum of a convex body and a ball.
For two convex bodies $X$ and $Y$ in $\mathbb{R}^d$, denote by $X+Y = \{x+y: x \in X \text{ and } y \in Y\}$ their Minkowski sum.

\begin{lemma} \label{lem-minkowski}
    For constant $d$, let $C \subseteq [0,1]^d$ be a convex body and $B_\delta$ be the ball with radius $\delta \in [0,1]$ centered at the origin of $\mathbb{R}^d$.
    Then $\lVert (C+ B_\delta) \backslash C \rVert = O(\delta)$ for all $\delta \in [0,1]$.
\end{lemma}
\begin{proof}
By Steiner's formula, $\lVert (C+ B_\delta) \rVert = \sum_{i=1}^d (W_i(C) \cdot \delta^i)$, where $W_i(C)$ is called the $i$-th \textit{quermassintegral} $C$.
The quermassintegral satisfies the following two conditions:
\smallskip
\begin{itemize}
    \item $W_0(X) = \lVert X \rVert$ for any convex body $X \subseteq \mathbb{R}^d$,
    \smallskip
    
    \item $W_i(X) \leq W_i(X')$ for any $i \in [d]$ and convex bodies $X,X' \subseteq \mathbb{R}^d$ satisfying $X \subseteq X'$.
    \smallskip
    
\end{itemize}
Therefore, $\lVert (C+ B_\delta) \rVert \leq \lVert C \rVert+ \rho \sum_{i=1}^d \delta^i$ where $\rho = \max_{i=1}^d W_i([0,1]^d)$ is a constant.
As $\delta \in [0,1]$, we have $\lVert (C+ B_\delta) \rVert = \lVert C \rVert+O(\delta)$.
\end{proof}

By construction, we have $\mathcal{CH}(P) \subseteq C^*$.
On the other hand, we observe that $C^* \subseteq \mathcal{CH}(P)+B$, where $B$ is the ball in $\mathbb{R}^d$ centered at the origin with radius $\frac{d}{r}$.
As $C^*$ is a polytope, it suffices to show that every vertex of $C^*$ lies in $\mathcal{CH}(P)+B$.
A vertex $z$ of $C^*$ is a corner of some $\Box \in \varGamma_1$.
Since $A_\Box = \mathsf{no}$, there exists a point $p \in P \cap \Box$.
The distance between $z$ and $p$ is at most $\frac{d}{r}$, as the side-length of $\Box$ is $\frac{1}{r}$.
Thus, $z \in p+B$ and $z \in \mathcal{CH}(P)+B$.
It follows that $\lVert C^* \backslash \mathcal{CH}(P) \rVert \leq \lVert (\mathcal{CH}(P)+B) \backslash \mathcal{CH}(P) \rVert$.
By Lemma~\ref{lem-minkowski}, $\lVert (\mathcal{CH}(P)+B) \backslash \mathcal{CH}(P) \rVert = O(\frac{1}{r})$.
So we have $\lVert C^* \backslash \mathcal{CH}(P) \rVert = O(\frac{1}{r}) = O(q^{-1/d})$.
\begin{theorem} \label{thm-nonadapt_orth_ub}
    There exists an algorithm for \textnormal{\sc Approximate Convex Hull} in $\mathbb{R}^d$ that performs $q$ non-adaptive orthogonal emptiness queries and has error $O(q^{-1/d})$.
\end{theorem}

Interestingly, the above algorithm, while being very simple, is already the best one can hope.
Specifically, we show that any deterministic algorithm for \textsc{Approximate Convex hull} with $q$ non-adaptive orthogonal emptiness queries has error $\Omega(q^{-1/d})$.

Consider such an algorithm $\mathbf{A}$.
%Let $\Box_1,\dots,\Box_q$ be the $q$ queries applied by $\mathbf{A}$.
Since the $q$ queries $\mathbf{A}$ performs are non-adaptive, these queries are independent of the point set $P$ (as well as the range-emptiness oracle).
Let $\Box_1,\dots,\Box_q$ be these queries, each of which is a box in $\mathbb{R}^d$.
Without loss of generality, we may assume $\Box_1,\dots,\Box_q \subseteq [0,1]^d$.
%Pick a sufficiently large number $c > 0$ that only depends on $d$.
Set $n = \lceil q^{1/d} \rceil +1$ and $\delta = q^{-1/d}/c$ where $c$ is a sufficiently large constant (which only depends on $d$).
Define a sequence $H_0,H_1,\dots,H_n$ of parallel hyperplanes in $\mathbb{R}^d$ where the equation of $H_i$ is $x_1 + \cdots + x_d = 1+\delta i$.
For $i \in [n]$, we say a box $\Box$ is \textit{$i$-good} if $H_i \cap \Box \neq \emptyset$ and $H_{i-1} \cap \Box = \emptyset$. See figure \ref{fig:nonadaptorth_lb}.

\begin{figure}[h!]
\centering
\begin{tikzpicture}[scale=7,>=stealth]

% PARAMETERS -----------------------------------------------------------
\def\t{0.1}
\def\a{0.6}
\def\b{0.42}
\def\w{0.35}
\def\h{0.18}

% UNIT SQUARE ----------------------------------------------------------
\draw[thick] (0,0) rectangle (1,1);

% Corner labels
\node[left]  at (0,1) {$(0,1)$};
\node[below] at (1,0) {$(1,0)$};

% PARALLEL DIAGONALS ---------------------------------------------------
\foreach \i in {1,2,5}{
   \pgfmathsetmacro{\s}{\i*\t}
   \draw[gray!70] (1,\s) -- (\s,1);
}
\draw[gray!70, dashed] (1,\t*3) -- (\t*3,1);
\draw[gray!70, dashed] (1,\t*4) -- (\t*4,1);

\draw[thick, black] (1, \t) -- (\t, 1);
\draw[ultra thick, blue] (\a, 1+ \t - \a) -- (1 + \t - \b, \b) node[blue, midway, right] {$H_1 \cap \Box$};

\draw[<->] (0.75, 0.25) -- (0.8, 0.3) node[midway, xshift=-3pt, yshift=3pt] {$\delta$};

% MAIN DIAGONAL --------------------------------------------------------
% Point set P in red
\draw[ultra thick, dotted, red] (1,0) -- (0,1);
\fill[red] (0.4, 0.7) circle (0.01);
\node[left, red, xshift=-2pt, yshift=-2pt] at (0.4, 0.7) {$P$};

% DIAGONAL LABELS ------------------------------------------------------
\foreach \i in {0,1,2}{
   \pgfmathsetmacro{\s}{\i*\t}
   \draw[gray!70] (1,\s) -- (\s,1)
        node[midway, above, xshift=-2pt] {$H_{\i}$};
}
\draw[gray!70] (1,\t*5) -- (\t*5,1)
        node[midway, above, xshift=-2pt] {$H_{n}$};

% RECTANGLE ------------------------------------------------------------
\draw[thick,black] (\a,\b) rectangle ++(\w,\h);

\end{tikzpicture}

\caption{The lower bound construction in $d = 2$ dimensions for non-adaptive orthogonal queries with underlying points $P$ (red) along the main diagonal $H_0$. The depicted query $\Box$ is $1$-good since its lower left corner lies between $H_0$ and $H_1$. The length of $H_1 \cap \Box$ (blue) is $O(\delta^{d-1})$, so $\Omega(1/\delta^{d-1})$ $1$-good queries are needed to cover $H_1$ in order to determine whether $P$ contains a point on $H_1$. In total, $\Omega(n\delta^{d-1})$ $i$-good queries are needed to cover all $H_1, \cdots, H_n$.}
\label{fig:nonadaptorth_lb}
\end{figure}
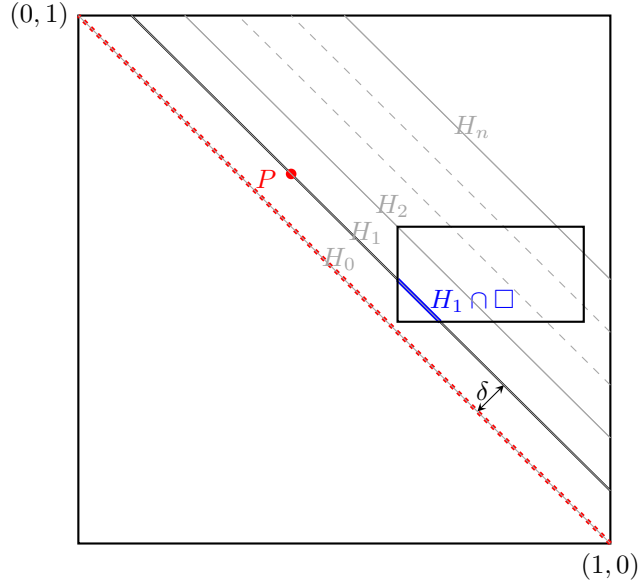

\begin{lemma}
    There exists $i \in [n]$ such that at most $\frac{q}{n}$ boxes in $\{\Box_1,\dots,\Box_q\}$ are $i$-good.
\end{lemma}
\begin{proof}
Observe that for different $i,j \in [n]$, no box in $\mathbb{R}^d$ can be both $i$-good and $j$-good.
To see this, suppose $i < j$ without loss of generality.
Let $\Box$ be a $j$-good box.
Then $H_j \cap \Box \neq \emptyset$ and $H_{j-1} \cap \Box = \emptyset$.
So $H_j$ and $\Box$ are on the same side of $H_{j-1}$.
If $i = j-1$, then $H_i \cap \Box = \emptyset$ and thus $\Box$ is not $i$-good.
Otherwise, $i \leq j-2$.
In this case, $H_i$ and $\Box$ are on the opposite sides of $H_{j-1}$, which implies $H_i \cap \Box = \emptyset$ and $\Box$ is not $i$-good.
Therefore, any box in $\mathbb{R}^d$ can be $i$-good for at most one $i \in [n]$.
Let $t_i$ be the number of $i$-good boxes in $\{\Box_1,\dots,\Box_q\}$.
It follows that $\sum_{i=1}^n t_i \leq q$ and hence there exists $i \in [n]$ satisfying $t_i \leq q/n$.
\end{proof}

% I removed that q/n = O(q^{1/d}); I believe q/n = q^(1-1/d)
Let $i \in [n]$ such that at most $\frac{q}{n}$ boxes in $\{\Box_1,\dots,\Box_q\}$ are $i$-good.
Without loss of generality, assume $\Box_1,\dots,\Box_t$ are $i$-good and $\Box_{t+1},\dots,\Box_q$ are not $i$-good, where $t \leq \frac{q}{n}$.
Each $H_i \cap \Box_j$ is a $(d-1)$-dimensional convex polytope, and let $\lVert H_i \cap \Box_j \rVert'$ denote its $(d-1)$-dimensional volume.

\begin{lemma} \label{lem-smallpiece}
    For all $j \in [t]$, $\lVert H_i \cap \Box_j \rVert' = O(\delta^{d-1})$.
\end{lemma}
\begin{proof}
Suppose $\Box_j = [x_1^-,x_1^+] \times \cdots \times [x_d^-,x_d^+]$ and $\Box_j' = [x_1^-,+\infty) \times \cdots \times [x_d^-,\infty)$.
Since $\Box_j$ is $i$-good, we have $H_i \cap \Box_j \neq \emptyset$ and $H_{i-1} \cap \Box_j = \emptyset$, which implies $1+\delta i -\delta \leq x_1^- +\cdots + x_d^- \leq 1+\delta i$.
So the distance from the point $(x_1^-,\dots,x_d^-) \in \mathbb{R}$ to $H_i$ is $O(\delta)$.
This implies $\lVert H_i \cap \Box_j' \rVert' = O(\delta^{d-1})$, and thus $\lVert H_i \cap \Box_j \rVert' = O(\delta^{d-1})$.
\end{proof}

\begin{lemma} \label{lem-avoidance}
    $H_i \cap [0,1]^d \nsubseteq \bigcup_{j=1}^t \Box_j$.
\end{lemma}
\begin{proof}
Equivalently, we prove that $H_i \cap [0,1]^d \nsubseteq \bigcup_{j=1}^t (H_i \cap \Box_j)$.
As $H_i \cap [0,1]^d$ and $H_i \cap \Box_1,\dots,H_i \cap \Box_t$ are all $(d-1)$-dimensional convex polytopes, it suffices to show that $\lVert H_i \cap [0,1]^d \rVert' > \sum_{j=1}^t \lVert H_i \cap \Box_j \rVert'$.
%Here $\lVert \cdot \rVert$ denotes the $(d-1)$-dimensional volume.
%Recall that the equation of $H_i$ is $x_1 + \cdots + x_d = 1+\delta i$.
% I weakened O(1/c * 1/q^{1/d}) to just O(1/c) here, to account for when i = q^{1/d}. Shouldn't change anything else.
Since $1+\delta i = 1+O(\frac{1}{c})$ and $c$ is sufficiently large, $\lVert H_i \cap [0,1]^d \rVert'$ is lower bounded by a constant depending only on $d$.
By Lemma~\ref{lem-smallpiece}, $\sum_{j=1}^t \lVert H_i \cap \Box_j \rVert' = O(\delta^{d-1} t) = O(\frac{1}{c})$.
Therefore, we have $\lVert H_i \cap [0,1]^d \rVert' > \sum_{j=1}^t \lVert H_i \cap \Box_j \rVert'$, for sufficiently large $c$.
\end{proof}

Next, we construct two sets $Z$ and $Z'$ of points in $\mathbb{R}^d$.
We include in $Z$ the vertices of the polytope $H_{i-1} \cap [0,1]^d$.
Furthermore, for each $j \in [q]$ such that $H_{i-1} \cap \Box_j \neq \emptyset$, we include in $Z$ a point $p_j \in H_{i-1} \cap \Box_j$.
Note that all points in $Z$ lie on $H_{i-1} \cap [0,1]^d$ and we have $\mathcal{CH}(Z) = H_{i-1} \cap [0,1]^d$, so $\|\mathcal{CH}(Z)\| = 0$.
To further construct $Z'$, we pick a point $p \in (H_i \cap [0,1]^d) \backslash (\bigcup_{j=1}^t \Box_j)$, which exists by Lemma~\ref{lem-avoidance}.
Then we define $Z' = Z \cup \{p\}$.
\begin{lemma}
    \label{lem-indistinguish}
    For all $j \in [q]$, $Z \cap \Box_j = \emptyset$ iff $Z' \cap \Box_j = \emptyset$.
\end{lemma}
\begin{proof}
If $p \notin \Box_j$, then $Z \cap \Box_j = Z' \cap \Box_j$ and we are done.
So assume $p \in \Box_j$.
Then $j \notin [t]$ by the choice of $p$.
Thus, $\Box_j$ is not $i$-good and we have either $H_i \cap \Box_j = \emptyset$ or $H_{i-1} \cap \Box_j \neq \emptyset$.
But $H_i \cap \Box_j \neq \emptyset$ since $p \in H_i \cap \Box_j$.
So we must have $H_{i-1} \cap \Box_j \neq \emptyset$.
By construction, we include in $Z$ the point $p_j \in H_{i-1} \cap \Box_j$, which implies $Z \cap \Box_j \neq \emptyset$ and $Z' \cap \Box_j \neq \emptyset$.
\end{proof}

The above lemma shows that the algorithm $\mathbf{A}$ cannot distinguish $Z$ and $Z'$, that is, it returns the same convex body when running on $P = Z$ and $P = Z'$.
\begin{lemma}
\label{lem-totalvol}
    $\lVert \mathcal{CH}(Z') \backslash \mathcal{CH}(Z) \rVert = \Omega(\delta)$.
\end{lemma}
\begin{proof}
As $\lVert \mathcal{CH}(Z) \rVert = 0$, it suffices to show $\lVert \mathcal{CH}(Z')\rVert = \Omega(\delta)$.
The distance from $p$ to $H_{i-1}$ is $\delta$, since $p \in H_i$.
Thus, $\lVert \mathcal{CH}(Z')\rVert = \Omega(\delta \cdot \lVert H_{i-1} \cap [0,1]^d  \rVert')$.
We have $\lVert H_{i-1} \cap [0,1]^d \rVert' = \Omega(1)$, because the equation of $H_{i-1}$ is $x_1+\cdots+x_d = 1+\delta(i-1)$ and $\delta(i-1) = O(\frac{1}{c})$.
It follows that $\lVert \mathcal{CH}(Z')\rVert = \Omega(\delta)$.
\end{proof}

Let $C^*$ be the output of $\mathbf{A}$ when running on $P = Z$ or $P = Z'$.
We have $\lVert \mathcal{CH}(Z') \backslash \mathcal{CH}(Z) \rVert \leq \lVert C^* \oplus \mathcal{CH}(Z) \rVert + \lVert C^* \oplus \mathcal{CH}(Z') \rVert$.
Therefore, the above lemma implies either $\lVert C^* \oplus \mathcal{CH}(Z) \rVert = \Omega(\delta)$ or $\lVert C^* \oplus \mathcal{CH}(Z') \rVert = \Omega(\delta)$.
As $\delta = \Theta(q^{-1/d})$, the error of $\mathbf{A}$ is $\Omega(q^{-1/d})$.

\begin{theorem}
\label{thm-nonadapt_orth_lb}
    Any deterministic algorithm for \textnormal{\sc Approximate Convex Hull} in $\mathbb{R}^d$ that performs $q$ non-adaptive orthogonal emptiness queries has error $\Omega(q^{-1/d})$.
\end{theorem}

\subsection{Adaptive orthogonal queries} \label{sec-adaptorth}
% TODO: move query definitions to preliminaries to avoid repetition?
%Let $P \subseteq [0, 1]^d$ be an (unknown) set of points. A rectangle query $\textsc{Query}(R)$ accepts a rectangle $R$, and returns the number of points in the rectangle $|P \cap R|$.
We now consider adaptive orthogonal emptiness queries, and show how to approximate $\mathcal{CH}(P)$ with error $O(q^{-1/(d-1)})$ via $O(q)$ queries.
We shall define $2^d$ supersets of $\mathcal{CH}(P)$, one-to-one corresponding to the vectors in $\{-1, 1\}^d$.
Our algorithm independently approximates these supersets and then merges them together.
%In two dimensions, these are the upper-right, upper-left, lower-right, and lower-left sections of the convex hull.
%For a $d$-dimensional point set, we define $2^d$ corners of the hull, one for each vector $v \in \{-1, 1\}^d$.
For each $\vec{v} = (v_1,\dots,v_d) \in \{-1, 1\}^d$, let $O_v$ be the orthant $\{(x_1, \dots, x_d) : x_i v_i\geq 0 \text{ for all } i \in [d]\}$.
We define $\mathcal{CH}_{\vec{v}} (P) = (\mathcal{CH}(P) + O_{-\vec{v}}) \cap [0,1]^d$.
%Intuitively, the boundary of $\mathcal{CH}_\vec{v} (P)$ contains all points on the convex hull of $P$ that are extremal in the direction of some vector in $O_{\vec{v}}$.
%It also contains the region of $[0,1]^d$ underneath this portion of the hull (in the sense that they can be reached by starting at a point on the boundary and traveling in the direction of some vector in $O^-_v$).
The following observation relates $\mathcal{CH}(P)$ with the supersets $\mathcal{CH}_{\vec{v}} (P)$.

%Given $\mathcal{CH}_{\vec{v}} (P)$ for all $\vec{v} \in \{-1, 1\}^d$, the following observation allows us to reconstruct $\mathcal{CH}(P)$.

\begin{lemma}
\label{lem-adaptrec}
$\mathcal{CH}(P) = \bigcap_{\vec{v} \in \{-1, 1\}^d} \mathcal{CH}_{\vec{v}}(P)$.
\end{lemma}
\begin{proof}
Consider an arbitrary point $p$.
By definition, if $p \in \mathcal{CH}(P)$, then $p \in \mathcal{CH}_{\vec{v}}(P)$ for all $\vec{v}$.
So assume $p \not\in \mathcal{CH}(P)$, and let $p' \in \partial\mathcal{CH}(P)$ be the closest point to $p$.
Then there is a hyperplane $H$ separating $p$ from $P$ that is perpendicular to the vector $p - p'$.
Let $\vec{v} \in \{-1, 1\}^d$ such that $p - p' \in O_{\vec{v}}$.
Then all points of $\mathcal{CH}_{\vec{v}}(P)$ are on the opposite side of $H$ from $p$, so we have $p \notin \mathcal{CH}_{\vec{v}}(P)$ and thus $p \notin \bigcap_{\vec{v} \in \{-1, 1\}^d} \mathcal{CH}_{\vec{v}}(P)$.
\end{proof}

Algorithm \ref{alg-AdaptiveRectangle} presents how to approximate $\mathcal{CH}_{\vec{v}} (P)$ for each $\vec{v} \in \{-1, 1\}^d$.
It essentially ``sandwiches'' the hull between two convex boundaries, refining the estimate at each iteration. Figure \ref{fig:adaptorth_alg} illustrates one such refinement step.
The subroutine $\textsc{Subdivide}(R)$ evenly partitions the hypercube $R$ into $2^d$ smaller hypercubes (each of which has side-length half of the side-length of $R$) and returns the set of these $2^d$ hypercubes.
%returns the $2^d$ pieces of $R$ produced by splitting $R$ exactly in half in each dimension (they coincide at the shared boundary).
For a hypercube $R$ and a vector $\vec{v} \in \{-1, 1\}^d$, we define $\mathsf{cor}(R, \vec{v}) = \arg \max_{p \in R} \langle p,\vec{v} \rangle$ as the corner of $R$ in the direction $\vec{v}$.
The algorithm starts with a set $\mathcal{R}$ of hypercubes which initially consists of only $[0,1]^d$, and runs in $\lceil \frac{\log q}{d-1}\rceil$ rounds.
In each round, we perform emptiness queries for the hypercubes in $\mathcal{R}$, and let $\mathcal{R'} \subseteq \mathcal{R}$ consists of the nonempty ones.
Then define $U = \mathcal{CH}_{\vec{v}}(\{\mathsf{cor}(R, \vec{v}): R \in \mathcal{R'}\})$ and $L' = \mathcal{CH}_{\vec{v}}(\{\mathsf{cor}(R, -\vec{v}): R \in \mathcal{R}'\})$.
Let $L = L' \setminus \partial L'$ be the interior of $L'$.
%(We remove the boundary of $L'$ so that $\partial \mathcal{CH}_{\vec{v}}$ is always ``sandwiched'' between $L$ and $U$.)
This is simply to ensure that $\partial \mathcal{CH}_{\vec{v}}(P)$ is contained in $U\setminus L$.
We then take all hypercubes in $\mathcal{R}$ that intersect $U \backslash L$, partition them using the subroutine $\textsc{Subdivide}$, and let the new $\mathcal{R}$ be the set of the resulting smaller hypercubes.
(Technically, we only need to consider the nonempty queries in $\mathcal{R}'$, but this is equivalent in the worst case, and this way nicely ensures that our queries continue to cover all of $\partial CH_{\vec{v}}(P)$.)
Finally, we return the convex body $U$ in the last iteration as an approximation of $\mathcal{CH}_{\vec{v}}(P)$.

\begin{algorithm}
    \caption{\textsc{AdaptiveOrthogonal}$(q, \vec{v})$}
    \begin{algorithmic}[1]
        \State $\mathcal{R} \leftarrow \{[0, 1]^d\}$
        \For{$i = 1$ to $\lceil \frac{\log q}{d-1}\rceil$}
            \State $\mathcal{R'} \leftarrow \{R\in \mathcal{R}:$ \textsc{Query}$(R) = \mathsf{no}\}$
            %\State $L \leftarrow \mathcal{CH}(\{p | R_{p, \cdot} \in \mathcal{R}\})$
            \State $U \leftarrow \mathcal{CH}_{\vec{v}}(\{\mathsf{cor}(R, \vec{v}): R \in \mathcal{R'}\})$
            \State $L' \leftarrow \mathcal{CH}_{\vec{v}}(\{\mathsf{cor}(R, -\vec{v}): R \in \mathcal{R'}\})$
            \State $L \leftarrow L' \setminus \partial L'$
            \State $\mathcal{R}'' \leftarrow \{R \in \mathcal{R}: R \cap (U \setminus L) \neq \emptyset\}$
            \State $\mathcal{R} \leftarrow \bigcup_{R \in  \mathcal{R''}} \textsc{Subdivide}(R)$
        \EndFor
    \State \textbf{return} $U$
    \end{algorithmic}
    \label{alg-AdaptiveRectangle}
\end{algorithm}

\begin{figure}[h!]
\centering
\begin{tikzpicture}[scale=7,>=stealth]

% UNIT SQUARE ---------------------------------------------------------
\draw[thick] (0,0) rectangle (1,1);
\node[left]  at (0,1) {$(0,1)$};
\node[below] at (1,0) {$(1,0)$};

% POINT SET P ---------------------------------------------------------
\coordinate (p1) at (0.8,0.3);
\coordinate (p2) at (0.6,0.7);
\coordinate (p3) at (0.4,0.6);
\coordinate (p4) at (0.2,0.2);
\coordinate (p5) at (0.3,0.9);
\coordinate (p6) at (0.65,0.55);

\coordinate (l1) at (0.25,0.75);
\coordinate (l2) at (0.25,0.5);
\coordinate (l3) at (0.5,0.5);
\coordinate (l4) at (0.75,0.25);
%\coordinate (l5) at (0.5,0.25);

\coordinate (u1) at (0.5,1);
\coordinate (u2) at (0.5,0.75);
\coordinate (u3) at (0.75,0.75);
\coordinate (u4) at (1,0.5);
%\coordinate (u5) at (0.75,0.5);

% LIST OF BOXES TO HIGHLIGHT
\foreach \i/\j in {2/4, 2/3, 3/3, 4/2} {
    % compute box coordinates
    \pgfmathsetmacro{\xmin}{(\i-1)*0.25}
    \pgfmathsetmacro{\xmax}{\i*0.25}
    \pgfmathsetmacro{\ymin}{(\j-1)*0.25}
    \pgfmathsetmacro{\ymax}{\j*0.25}
    
    % draw highlighted rectangle
    \draw[black!60, very thick, fill=gray!20] (\xmin,\ymin) rectangle (\xmax,\ymax);
}

% QUAD TREE DISSECTION -----------------------------------------------
% Split the square into 4 (first level)
\draw[gray!60, dashed] (0.5,0) -- (0.5,1);
\draw[gray!60, dashed] (0,0.5) -- (1,0.5);

% Split top-right quadrant further (example)
\draw[gray!60, dashed] (0.75,0.5) -- (0.75,1);
\draw[gray!60, dashed] (0.5,0.75) -- (1,0.75);

% Split top-left quadrant further (example)
\draw[gray!60, dashed] (0.25,0.5) -- (0.25,1);
\draw[gray!60, dashed] (0,0.75) -- (0.5,0.75);

% Split bottom-right quadrant further (example)
\draw[gray!60, dashed] (0.75,0) -- (0.75,0.5);
\draw[gray!60, dashed] (1,0.25) -- (0.5,0.25);

\foreach \i/\j in {
    4/1, 3/2, 4/2, 2/3, 3/3, 4/3, 1/4, 2/4, 3/4
} {
    % original box coordinates
    \pgfmathsetmacro{\xmin}{(\i-1)*0.25}
    \pgfmathsetmacro{\xmax}{\i*0.25}
    \pgfmathsetmacro{\ymin}{(\j-1)*0.25}
    \pgfmathsetmacro{\ymax}{\j*0.25}
    
    % size of sub-box
    \pgfmathsetmacro{\dx}{(\xmax-\xmin)/2}
    \pgfmathsetmacro{\dy}{(\ymax-\ymin)/2}
    
    % draw sub-boxes
    \foreach \ii in {0,1} {
        \foreach \jj in {0,1} {
            \draw[gray, dashed] 
                ({\xmin+\ii*\dx},{\ymin+\jj*\dy}) rectangle 
                ({\xmin+(\ii+1)*\dx},{\ymin+(\jj+1)*\dy});
        }
    }
}

% Draw points
\foreach \p in {p1,p2,p3,p4,p5,p6}{
    \fill[black] (\p) circle (0.015);
}

\foreach \p in {l1,l2,l3,l4}{
    \fill[red] (\p) circle (0.015);
}
\foreach \p in {u1,u2,u3,u4}{
    \fill[blue] (\p) circle (0.015);
}

% CONVEX HULL ---------------------------------------------------------
\draw[thick, black, rounded corners=2pt] 
    (p1) -- (p2) -- (p5);
\draw[thick, black, rounded corners=2pt]
    (p1) -- (0.8,0);
\draw[thick, black, rounded corners=2pt]
    (p5) -- (0,0.9);

\draw[thick, blue, rounded corners=2pt] 
    (u1) -- (u3) -- (u4);
\draw[thick, blue, rounded corners=2pt]
    (u1) -- (0,1);
\draw[thick, blue, rounded corners=2pt]
    (u4) -- (1,0);

\draw[thick, red, rounded corners=2pt] 
    (l1) -- (l3) -- (l4);
\draw[thick, red, rounded corners=2pt]
    (l1) -- (0,0.75);
\draw[thick, red, rounded corners=2pt]
    (l4) -- (0.75,0);

% Remaining labels
\node[above right, blue] at (0.875,0.625) {${\partial U}$};
\node[below left, red] at (0.625,0.375) {${\partial L}$};
\node[above] at (0.875,0.35) {${\partial \mathcal{CH}_{\vec{v}}(P)}$};

\end{tikzpicture}

\caption{One iteration of Algorithm \ref{alg-AdaptiveRectangle} for $d=2, \vec{v}=(1,1)$.
The four queries in $\mathcal{R'}$ are filled in gray.
The upper right corners of these queries are used to compute the set $U$ with boundary shown in blue, while the lower left corners produce $L$ in red. The boundary of $\mathcal{CH}_{\vec{v}}(P)$ lies in $U\setminus L$. The smallest dashed queries form the set $\mathcal{R}$ for the next iteration; their union covers $U \setminus L$.}
\label{fig:adaptorth_alg}
\end{figure}

% TODO: It's annoying to have to subtract the boundary of L' at line 7, but this makes the claim that the boundary of CH_v(P) contained in U - L possible (otherwise it would be U - L + boundary L), which seems equally annoying. Maybe there's a direct way to handle this without the second inductive  claim in the proof.
At a high level, the algorithm maintains an upper and lower bound (the sets $U$ and $L$ respectively) on the unknown convex body $\mathcal{CH}_{\vec{v}}(P)$.
At each iteration of the loop, we reduce the error volume $\lVert (U \backslash \mathcal{CH}_{\vec{v}}(P)) \cap [0,1]^d \rVert$ by half.
Let $t = \lceil \frac{\log q}{d-1}\rceil$ be the number of iterations.
For $i \in [t]$, denote by $\mathcal{R}_i$ the set $\mathcal{R}$ at the end of the $i$-th iteration of the loop, and likewise for $L_i$, $U_i$, $\mathcal{R}'_i$, $\mathcal{R}_i''$.
Without loss of generality, we assume that $\vec{v} = (1,\dots,1)$.
The following two lemmas are used to bound the error and number of queries.%; their proofs are included in the full version.

\begin{lemma}
\label{lem-adaptiverec_bound}
For all $i \in [t]$, $L_i \subset \mathcal{CH}_{\vec{v}}(P) \subseteq U_i$ and $\partial \mathcal{CH}_{\vec{v}}(P) \subseteq \bigcup_{R \in \mathcal{R}_i}R$.
\end{lemma}
\begin{proof}
%It suffices to show that for each $i$, $\partial \mathcal{CH}_v(P) \subseteq U_i \setminus L_i$.
We prove both statements simultaneously by induction on $i$.
Initially $\mathcal{R}$ contains a single rectangle $R_0 = [0, 1]^d$ such that $P \subseteq R_0$ and $\partial \mathcal{CH}_{\vec{v}}(P) \subseteq R_0$ by definition.
In the first iteration, $\mathcal{R'}_1 = \{R_0\}$, $U_1 = [0,1]^d$, and $L_1 = \emptyset$.
Clearly, $L_1 \subset \mathcal{CH}_{\vec{v}}(P) \subseteq U_1$ for any nonempty set $P \subseteq [0,1]^d$.
%Also $\partial \mathcal{CH}_v(P) \subseteq R_0$, so $\partial \mathcal{CH}_v(P) \subseteq \bigcup_{R \in \mathcal{R}_1} R$ since $\bigcup_{R \in \mathcal{R}_1} R = \bigcup_{R \in \textsc{Subdivide}(R_0)} R = R_0$.
Assume that $L_{i-1} \subset \mathcal{CH}_{\vec{v}}(P) \subseteq U_{i-1}$ and $\partial\mathcal{CH}_{\vec{v}}(P) \subseteq \bigcup_{R \in \mathcal{R}_{i-1}}R$; we will prove both statements for index $i$.
Fix some $p \in \partial \mathcal{CH}_{\vec{v}}(P)$.
By the second inductive claim, there exists $R_p \in \mathcal{R}_{i-1}$ such that $p \in R_p$.
Then $\textsc{Query}(R_p) =$ no, so $R_p \in \mathcal{R'}_i$.
Let $u = \mathsf{cor}(R_p, {\vec{v}})$.
By definition, $u \in U$.
Furthermore, $(p - u) \in O_{-{\vec{v}}}$ since $(p_j - u_j) \cdot v_j \leq 0$ for all $j \in [d]$, where $v_j$ is the $j$-th coordinate of $\vec{v}$.
Then $p = u + (p - u) \in U_i$, so $\partial \mathcal{CH}_{\vec{v}}(P) \subseteq U_i$, which implies $\mathcal{CH}_{\vec{v}}(P) \subseteq U_i$.
Similarly, for each $p \in \{\mathsf{cor}(R, -{\vec{v}}): R \in \mathcal{R}_i'\}$, there exists $p' \in P$ with $p' - p \in O_{-{\vec{v}}}$, witnessing that $p \in \mathcal{CH}_{{\vec{v}}}(P)$.
Then $L_i \subseteq \mathcal{CH}_{{\vec{v}}}(P)$.

Finally, we show that $\partial \mathcal{CH}_{\vec{v}}(P) \subseteq \bigcup_{R \in \mathcal{R}_i}R$.
% TODO: avoid repetion here
Again, fix $p \in \partial \mathcal{CH}_{\vec{v}}(P)$. 
There exists $R_p \in \mathcal{R}_{i-1}$ such that $p \in R_p$.
Then $R_p \in R'_{i}$ since $\textsc{Query}(R_p) = \mathsf{no}$.
Since we have already shown $L_i \subset \mathcal{CH}_{\vec{v}}(P) \subseteq U_i$, it must be that $\partial \mathcal{CH}_{\vec{v}}(P) \subseteq U_i \setminus L_i$, which implies $p \in U_i \setminus L_i$.
$R_p$ intersects $U_i \setminus L_i$ at point $p$, so $R_p \in \mathcal{R}''_i$.
Then there exists $R'_p \in \textsc{Subdivide}(R_p)$ such that $p \in R'_p$ because $\bigcup_{R\in \textsc{Subdivide}(R_p)} R = R_p$.
$R_p' \in \mathcal{R}_i$, so $p \in \bigcup_{R \in \mathcal{R}_i} R$, and therefore $\partial \mathcal{CH}_{\vec{v}}(P) \subseteq \bigcup_{R \in \mathcal{R}_i} R$. 
\end{proof}

\begin{lemma}
\label{lem-adaptrec_queries}
The number of queries at the $i$-th iteration is at most $4^dd!\cdot(2^i + 1)^{d-1}$.
\end{lemma}
\begin{proof}
Fix some $i \in [t]$, and let $r = 2^i$.
It suffices to show that $|\mathcal{R''}_i| = 2^dd! (r + 1) ^{d-1}$, since $\mathcal{R}_i$ contains exactly $2^d$ rectangles for each $R \in \mathcal{R}''_i$.
For all $R \in \mathcal{R}'_i$, $R \cap (U_i \setminus L_i) \neq \emptyset$ if and only if $R \cap \partial U_i \neq \emptyset$ or $R \cap \partial L_i \neq \emptyset$.
We bound both $|\{R \in \mathcal{R}': R \cap \partial U_i \neq \emptyset\}|$, $|\{R \in \mathcal{R}': R \cap \partial L_i \neq \emptyset\}| \leq 2^{d+1}d! (r+1)^{d-1}$.
Specifically, we prove that for any $d$-dimensional convex body $C$ and set of $d$-dimensional rectangles $\mathcal{R} = \mathcal{R}'_i$ produced by Algorithm \ref{alg-AdaptiveRectangle}, $|\{R \in \mathcal{R}: R \cap \partial C \neq \emptyset\}| \leq 2^{d+1}d! (r+1)^{d-1}$.
For $d=1$, $\partial C$ contains at most two points (the endpoints of a closed interval) each of which can intersect at most two intervals $R \in \mathcal{R}$.
Then for $1$-dimensional $C$ and all $i \in [t]$, $|\{R \in \mathcal{R}_1: R \cap \partial C \neq \emptyset| \leq 4$.

We now induct on $d$.
For $d > 1$, notice that if $R \in \mathcal{R}$ intersects $\partial C$, either $\partial R$ also intersects $\partial C$, or $C \subset R$.
In the latter case, $|\{R \in \mathcal{R}: R \cap \partial C\}| = 1$ since the rectangles of $\mathcal{R}$ are pairwise interior-disjoint, so we may assume $\partial R \cap \partial C \neq \emptyset$ for all $R$ such that $R \cap \partial C \neq \emptyset$.
Furthermore, across all $R \in \mathcal{R}$, the bounds of $R$ use at most $r+1$ unique coordinates in each dimension.
Indeed, let $H_{\alpha,\beta}$ be the $(d-1)$-dimensional hyperplane $\{x: x_j = \alpha/r\}$ for $\alpha \in [r] \cup \{0\}$ and $\beta \in [d]$, and $\mathcal{H} = \{H_{\alpha,\beta} : \alpha \in [r]\cup\{0\}, \beta \in [d]\}$.
Then $\bigcup_{R \in \mathcal{R}} \partial R \subseteq \bigcup_{H \in \mathcal{H}} H$.
We may write $\{R \in \mathcal{R}: \partial R \cap \partial C \neq \emptyset\} = \bigcup_{H \in \mathcal{H}}\{R \in \mathcal{R}: \partial R \cap \partial C \cap H \neq \emptyset\}$.
But now for all $H \in \mathcal{H}$, $\partial C \cap H$ is the boundary of some $(d-1)$-dimensional convex body, and $\mathcal{R}^{d-1} = \{R \cap H \neq \emptyset: R \in \mathcal{R}\}$ is precisely the set $\mathcal{R}'_i$ produced by running Algorithm \ref{alg-AdaptiveRectangle} in dimension $d-1$ on the projection of the point set $P$ onto $H$.
(Note, however, that $|\{R \in \mathcal{R}: R \cap H \neq \emptyset\}| = 2 \cdot |\{R \cap H \neq \emptyset: R \in \mathcal{R}\}|$, since each $(d-1)$-dimensional rectangle in the former maps to two $d$-dimensional rectangles in the latter, one on each side of $H$.)
By induction, $|\{R \in \mathcal{R}: \partial R \cap \partial C \cap H \neq \emptyset\}| \leq 2|\{R \in \mathcal{R}^{d-1}: R \cap (\partial C \cap H) \neq \emptyset\}| \leq 2\cdot 2^{d}(d-1)! (r+1)^{d-2})$ for all $H \in \mathcal{H}$.
$|\mathcal{H}| = d(r+1)$, so, by union bound,
$|\bigcup_{H \in \mathcal{H}}\{R\in \mathcal{R}: \partial R \cap \partial C \cap H \neq \emptyset\}| \leq 2^{d+1}(d-1)! (r+1)^{d-2} \cdot d(r+1) = 2^{d+1}d!(r+1)^{d-1}$.
\end{proof}

As $d$ is a constant, Lemma~\ref{lem-adaptrec_queries} implies the number of queries in the $i$-th iteration is $O((2^i+1)^{d-1})$.
So the total number of queries is bounded by $\sum_{i=1}^t (2^i+1)^{d-1} = O(q)$.
Since each query box in the $i$-th iteration has volume $2^{-id}$, we have 
\begin{equation*}
    \left\lVert \bigcup_{R\in \mathcal{R}_t}R \right\rVert \leq 2^{-td} |\mathcal{R}_t| = O(2^{-d\log q / (d-1)}q) = O(q^{-1/(d-1)}).
\end{equation*}
By Lemma~\ref{lem-adaptiverec_bound}, this further implies $\lVert (U_t \backslash \mathcal{CH}_{\vec{v}}(P)) \cap [0,1]^d \rVert = O(q^{-1/(d-1)})$, simply because $(U_t \backslash \mathcal{CH}_{\vec{v}}(P)) \cap [0,1]^d \subseteq (U_t \backslash L_t) \cap [0,1]^d \subseteq \bigcup_{R\in \mathcal{R}_t}R$.

We run Algorithm~\ref{alg-AdaptiveRectangle} for all $\vec{v} \in \{-1, 1\}^d$ and let $U_{\vec{v}}$ be the output for $\vec{v}$.
Finally, we return $C^* = (\bigcap_{\vec{v} \in \{-1, 1\}^d} U_{\vec{v}}) \cap [0,1]^d$ as the approximation of $\mathcal{CH}(P)$.
The total number of queries is still $O(q)$.
To bound the error $\lVert C^* \oplus \mathcal{CH}(P) \rVert$, we first observe $\mathcal{CH}(P) \subseteq C^*$.
Indeed, since $\mathcal{CH}_{\vec{v}}(P) \subseteq U_{\vec{v}}$ for all $\vec{v} \in \{-1, 1\}^d$, we have $\mathcal{CH}(P) \subseteq \bigcap_{\vec{v} \in \{-1, 1\}^d} U_{\vec{v}}$ by Lemma~\ref{lem-adaptrec}, which implies $\mathcal{CH}(P) \subseteq C^*$ because $P \subseteq [0,1]^d$.
Again, by Lemma~\ref{lem-adaptrec}, we have
\begin{equation*}
    \lVert C^* \backslash \mathcal{CH}(P) \rVert \leq \sum_{\vec{v} \in \{-1, 1\}^d} \lVert (U_{\vec{v}} \backslash \mathcal{CH}_{\vec{v}}(P)) \cap [0,1]^d \rVert = O(q^{-1/(d-1)}).
\end{equation*}
Therefore, the error of our algorithm is $O(q^{-1/(d-1)})$.

%a total of $2^d$ times to approximate the corner $U_v \approx \mathcal{CH}_v$ for each $v \in \{-1, 1\}^d$.
%We then produce the final approximation using Observation \ref{obs-adaptrec}, i.e. we take the intersection of the $2^d$ output sets $U_v$.
%The total error volume is at most $2^d$ times $\max_v ||U_v\setminus L_v||$, and the $2^d$ factor is hidden by $O(\cdot)$ for constant $d$.
%We get the following theorem.

\begin{theorem} \label{thm-adapt_orth_ub}
There exists an algorithm for \textnormal{\sc Approximate Convex Hull} in $\mathbb{R}^d$ that performs $O(q)$ adaptive orthogonal emptiness queries and has error $O(q^{-1/(d-1)})$.
\end{theorem}

We complement this result by showing that any deterministic algorithm with $q$ adaptive orthogonal queries has error $\Omega(q^{-1/(d-1)})$.
The argument is similar to and simpler than that of Theorem \ref{thm-nonadapt_orth_lb}.
%so we include a detailed proof in the full version and provide a brief discussion below.
Consider a deterministic algorithm $\mathbf{A}$ that performs $q$ adaptive orthogonal queries.
The idea is again to construct two sets $Z$ and $Z'$ which $\mathbf{A}$ cannot distinguish.
Take two parallel hyperplanes $H$ and $H'$ in $\mathbb{R}^d$ with distance $\delta = q^{-1/(d-1)}/c$ for a sufficiently large constant $c$ such that the $(d-1)$-dimensional volumes of $H \cap [0,1]^d$ and $H' \cap [0,1]^d$ are $\Omega(1)$.
Then $Z$ is essentially the set of all points in $H \cap \mathbb{R}^d$.
We run $\mathbf{A}$ on $P = Z$, and let $\Box_1,\dots,\Box_q$ be the queries performed.
We say $\Box$ is \textit{good} if $\Box \cap H = \emptyset$ and $\Box \cap H' \neq \emptyset$.
By construction, if $\Box$ is good, we can guarantee that the $(d-1)$-dimensional volume of $\Box \cap H'$ is way smaller than $\frac{1}{q}$.
Thus, the good boxes among $\Box_1,\dots,\Box_q$ cannot cover $H' \cap [0,1]^d$, implying the existence of a point $p^* \in H' \cap [0,1]^d$ not contained in any of $\Box_1,\dots,\Box_q$.
Then simply define $Z' = Z \cup \{p^*\}$.
It turns out that $\Box_i \cap Z = \emptyset$ iff $\Box_i \cap Z' = \emptyset$ for all $i \in [q]$.
As such, when running on $P = Z$ and $P = Z'$, the algorithm $\mathbf{A}$ performs exactly the same and provides the same output.
The fact $\lVert \mathcal{CH}(Z') \backslash \mathcal{CH}(Z) \rVert = \Omega(\delta) = \Omega(q^{-1/(d-1)})$ then implies that the error of $\mathbf{A}$ is $\Omega(q^{-1/(d-1)})$.

%We produce two sets of points $Z$, $Z'$, and similarly argue that any deterministic algorithm needs to make many good queries in order to distinguish $Z$ from $Z'$.
%However, there are two key differences to cope with the fact that these queries may depend on our chosen point set.
%First, we start with only two hyperplanes and have to assume that all queries can be good.
%Second, we have to make $Z$ a sufficiently dense set of points to ensure that each not good query contains a point of $Z$, since we cannot pick one point from each query a priori. The proof is similar to that of Theorem \ref{thm-nonadapt_orth_lb}, and included in the appendix.

\begin{theorem} \label{thm-adapt_orth_lb}
%For any constant $d \geq 1$, there exist $q > 0$ and $P_A \subset \mathbb{R}^d$ such that $A$ approximates $||\mathcal{CH}(P_A)||_d$ with additive error $O(1/q)$ only if $A$ uses $\Omega(q^{d-1})$ adaptive rectangle queries.
Any deterministic algorithm for \textnormal{\sc Approximate Convex Hull} in $\mathbb{R}^d$ that performs $q$ adaptive orthogonal emptiness queries has error $\Omega(q^{-1/(d-1)})$.
\end{theorem}

\begin{proof}
Let $\mathbf{A}$ be any deterministic algorithm that, given $q \geq 0$, $d \geq 1$, and oracle access $\mathcal{O_P}$ to an unknown point set $P$, approximates $\lVert \mathcal{CH}(P) \rVert$ using at most $q$ adaptive orthogonal emptiness queries.
We use just two of the hyperplanes defined in the previous section.
In particular, we have $H_0$, $H_1$, where $H_i$ is defined by the equations $x_1 + \cdots + x_d = 1 + \delta i$. 
We similarly define a query to be \textit{good} if it intersects $H_1$ but not $H_0$ (consistent with the definition of $1$-good).
Different from the previous section, we let $\delta = 1/q^{d-1}$.

Now, let $H_0'$ be the hyperplane defined by $x_1 + \cdots + x_d = 1 + \epsilon$ for some $\epsilon << \delta$.
Let $Z$ be a sufficiently dense set of points on $H_0'$ such that any orthogonal range intersecting both $H_0$ and $H_1$ includes at least one point of $Z$ (this density will depend on $\epsilon$).
We run $\mathbf{A}$ on $P$ and let $\Box_1, \cdots, \Box_q$ be its queries, labeled such that $\Box_1, \cdots, \Box_t$ are good and $\Box_{t+1}, \cdots, \Box_q$ are not, for some $0 \leq t \leq q$ determined by $\mathbf{A}$.
Now that $\mathbf{A}$'s queries on $Z$ are fixed, we may apply a similar argument to construct $Z'$.
By Lemma \ref{lem-smallpiece}, $\lVert H_i \cap \Box_j \rVert = O(\delta^{d-1})$ for $j \in [t]$.
Now we would like to apply Lemma \ref{lem-avoidance}, but we do not have the needed bound on $t$ (recall in the proof of Lemma \ref{lem-avoidance} that $t \leq q/n$ for $n = \Theta(q^{1/d})$); we only have $t \leq q$.
However, we have conveniently chosen a smaller $\delta = 1/(qc)^{d-1}$ so that Lemma \ref{lem-avoidance} still holds.
Indeed, we still have that $\lVert H_1 \cap [0,1]^d \rVert = \Omega(1)$, while $\sum_{j=1}^t ||H_i \cap \Box_j|| = O(\delta^{d-1}t) = O(1/(qc) \cdot q) = O(1/c)$.
As a consequence, we may select a point $p \in (H_1 \cap [0,1]^d)\setminus (\bigcup_{j=1}^t \Box_j)$, and similarly define $Z' = Z \cup \{p\}$.
We have preserved the relevant properties of not good queries so that Lemma \ref{lem-indistinguish} holds. 
In particular, $\mathbf{A}$ cannot distinguish $Z$ from $Z'$.
Finally, we compute the volume of $\mathcal{CH}(Z') \setminus \mathcal{CH}(Z).$
We use the same argument as in Lemma \ref{lem-totalvol}, with the caveat that the distance from $p$ to $H_0'$ (the hyperplane containing $Z$) is only $\delta - \epsilon$.
However, for a sufficiently small choice of $\epsilon << \delta$, we get the desired bound
$\lVert \mathcal{CH}(Z') \setminus \mathcal{CH}(Z) \rVert = \Omega(\delta - \epsilon) = \Omega(1/q^{d-1})$.
Thus, for one of $Z, Z'$, it must be that $\mathbf{A}$ answers with error $\Omega(1/q^{d-1})$.
\end{proof}

\section{Approximating convex hulls via halfplane queries}

In this section, we present our results for \textsc{Approximate Convex Hull} in $\mathbb{R}^2$ via halfplane emptiness queries.
%For non-adaptive queries (Section~\ref{sec-nonadaptorth}), the algorithm is very simple, while the lower bound proof is nontrivial and interesting.
%For adaptive queries (Section~\ref{sec-adaptorth}), the algorithm is more technical and the lower bound proof is simpler.
%In this section, we develop algorithms and lower bounds for approximating convex hulls via halfplane emptiness queries.
%Similar to the previous section, we can only access the point set $P$ via an oracle $O$ that accepts a query region $H$, and returns yes if $P \cap H = \emptyset$, no otherwise.
%The only difference is that $H$ is now a halfplane.
Instead of working on halfplane emptiness queries directly, it is more convenient to consider \textit{extreme halfplane queries}, which we formalize below.
For two parallel lines $\ell$ and $\ell'$, let $\mathsf{dist}(\ell,\ell')$ denote the distance between them.

\begin{definition}[extreme halfplane oracle]
    An \textbf{extreme halfplane oracle} on $P$ takes a unit vector $\vec{v} \in \mathbb{S}^1$ as input and returns a halfplane $H$ with normal vector $\vec{v}$ such that $P \subseteq H$.
    We say the oracle is \textbf{$\delta$-accurate} for a number $\delta \geq 0$ if for every $\vec{v} \in \mathbb{S}^1$, the halfplane $H$ returned by the oracle for query $\vec{v}$ satisfies the condition $\mathsf{dist}(\partial H, \partial H^*) \leq \delta$, where $H^*$ is the minimal halfplane with normal vector $\vec{v}$ satisfying $P \subseteq H^*$.
\end{definition}

It is easy to see that a $0$-accurate extreme halfplane oracle is stronger than a halfplane emptiness oracle.
On the other hand, one can simulate an extreme halfplane query with certain accuracy via multiple non-adaptive or adaptive halfplane emptiness queries.

\begin{lemma} \label{lem-extremetoemptiness}
    Given access to a halfplane emptiness oracle $\mathcal{O}$ on an unknown set $P$ of points in the plane, one can build a $\delta$-accurate extreme halfplane oracle on $P$ which performs $O(\frac{1}{\delta})$ non-adaptive queries to $\mathcal{O}$ or performs $O(\log \frac{1}{\delta})$ adaptive queries to $\mathcal{O}$.
\end{lemma}
\begin{proof}
Let $\vec{v} \in \mathbb{S}^1$ be the query vector of our extreme halfplane oracle.
We first consider how to build an extreme halfplane oracle on $P$ by applying non-adaptive queries to $\mathcal{O}$.
We construct a sequence of nested halfplanes $H_0 \subseteq H_1 \subseteq \cdots \subseteq H_n$ with normal vector $-\vec{v}$ such that $H_0 \cap [0,1]^2 = \emptyset$, $[0,1]^2 \subseteq H_n$, and $\mathsf{dist}(\partial H_{i-1}, \partial H_i) = \delta$ for all $i \in [n]$.
Clearly, such a sequence exists for $n = O(\frac{1}{\delta})$.
We simply query $\mathcal{O}$ using $H_0,H_1,\dots,H_n$ to find the largest $i \in [n]$ such that $P \cap H_i = \emptyset$.
Then our extreme halfplane oracle returns $H_i'$, which is the halfspace with normal vector $\vec{v}$ satisfying $\partial H_i' = \partial H_i$.
Since $P \cap H_i = \emptyset$, we have $P \subseteq H_i'$.
Furthermore, since $P \cap H_{i+1} = \emptyset$, the oracle is $\delta$-accurate.

If we are allowed to apply adaptive queries to $\mathcal{O}$, then we can reduce the number of queries by binary search.
As above, we try to find the largest $i \in [n]$ such that $P \cap H_i = \emptyset$ and then return $H_i'$.
With adaptive queries, we can do binary search in the sequence $H_0,H_1,\dots,H_n$ to find $i$, which only requires $\log n = O(\log \frac{1}{\delta})$ queries to $\mathcal{O}$.
\end{proof}

We will also need the following lemma to argue that the error introduced by querying a $\delta$-accurate extreme halfplane oracle, as opposed to a perfectly accurate one, is only $O(\delta)$.
%The proof is included in the full version.

\begin{lemma} \label{lem-enlarge}
    Let $H_1,\dots,H_r$ be halfplanes, $C = \bigcap_{i=1}^r H_i$, and $\delta \in [0,1]$.
    Suppose $C \subseteq [0,1]^2$.
    For each $i \in [r]$, let $H_i'$ be a halfplane parallel to $H_i$ such that $C \subseteq H_i$ and $\mathsf{dist}(\partial H_i', \partial H_i) \leq \delta$.
    Define $C' = \bigcap_{i=1}^r H_i'$.
    Then $\lVert (C' \cap [0,1]^2) \oplus C \rVert = O(\delta)$.
\end{lemma}
\begin{proof}
    Without loss of generality, we can assume that $C \cap \partial H_i \neq \emptyset$ for all $i \in [r]$ and the normal vectors of $H_1,\dots,H_r$ are sorted in clockwise order.
Let $\sigma_i$ be the edge of $C$ corresponding to $H_i$ for $i \in [r]$.
Without loss of generality, assume $\sigma_1,\dots,\sigma_r$ are sorted in clockwise order along $\partial C$.
For convenience, write $\sigma_0 = \sigma_r$.
For $i \in [r]$, let $a_i = \sigma_{i-1} \cap \sigma_i$ be the vertex of $C$ incident to $\sigma_{i-1}$ and $\sigma_i$.
Similarly, let $\sigma_i'$ be the edge of $C'$ corresponding to $H_i'$ and $a_i'= \sigma_{i-1}' \cap \sigma_i'$ for $i \in [r]$.
Define $s_i$ as the segment connecting $a_i$ and $a_i'$.
Then for each $i \in [r]$, the four segments $\sigma_i$, $s_i$, $\sigma_i'$, and $s_{i+1}$ enclose a trapezoid $T_i$.
The height of $T_i$ is $\delta$ and the two parallel edges of $T_i$ are $\sigma_i$ and $\sigma_i'$, which implies $\lVert T_i \rVert = (|\sigma_i|+|\sigma_i'|) \cdot \delta$.
Note that $(C' \cap [0,1]^2) \oplus C = (C' \cap [0,1]^2) \backslash C = \bigcup_{i=1} (T_i \cap [0,1]^2)$.
%So it suffices to bound $\sum_{i \in I} \lVert T_i \cap [0,1]^2 \rVert$.
%We claim that $\lVert T_i \cap [0,1]^2 \rVert = O(|\sigma_i' \cap [0,1]^2| \cdot \delta)$.

Let $I = \{i \in [r]: \sigma_i' \subseteq [-\delta,1+\delta]^2\}$.
%We have $\lVert T_i \cap [0,1]^2 \rVert = \lVert T_i \rVert = (|\sigma_i|+|\sigma_i'|) \cdot \delta$ for all $i \in I$.
Then $\sigma_i'$ is an edge of the convex polygon $C' \cap [-\delta,1+\delta]^2$ for all $i \in I$.
Therefore, we have $\sum_{i \in I} |\sigma_i'| = O(1+\delta) = O(1)$.
%Then $\sum_{i \in I} \lVert T_i \rVert = O(\delta \cdot \sum_{i \in I} |\sigma_i|) = O(\delta)$.
It follows that
\begin{equation*}
    \sum_{i \in I} \lVert T_i \cap [0,1]^2 \rVert \leq \sum_{i \in I} \lVert T_i \rVert = \left( \sum_{i \in I} |\sigma_i| + \sum_{i \in I} |\sigma_i'| \right) \delta = O(\delta).
\end{equation*}
Next, we observe $|[r] \backslash I| = O(1)$.
By our construction, it is easy to see that $\sigma_i' \cap [-\delta,1+\delta]^2 \neq \emptyset$ for all $i \in [r]$.
Thus, for each $i \in [r] \backslash I$, $\sigma_i'$ intersects the boundary of $[-\delta,1+\delta]^2$.
But the boundary of $[-\delta,1+\delta]^2$ consists of four edges, each of which can intersect at most two $\sigma_i'$, which implies $|[r] \backslash I| \leq 8$.
To complete the proof, it suffices to show $\lVert T_i \cap [0,1]^2 \rVert = O(\delta)$ for all $i \in [r] \backslash I$.
Let $\ell$ (resp., $\ell'$) be the line that contains $\sigma_i$ (resp., $\sigma_i'$), and $z$ be the strip in between $\ell$ and $\ell'$.
So the width of $z$ is $\delta$ and $\lVert z \cap [0,1]^2 \rVert = O(\delta)$.
Therefore, we have $\lVert T_i \cap [0,1]^2 \rVert \leq \lVert z \cap [0,1]^2 \rVert = O(\delta)$.
As $|[r] \backslash I| = O(1)$, we further deduce that $\sum_{i \in [r] \backslash I} \lVert T_i \cap [0,1]^2 \rVert = O(\delta)$.
Finally, the fact $\lVert (C' \cap [0,1]^2) \oplus C \rVert = \sum_{i=1}^r  \lVert T_i \cap [0,1]^2 \rVert$ implies the desired bound $\lVert (C' \cap [0,1]^2) \oplus C \rVert = O(\delta)$.
\end{proof}

In what follows, we discuss the results for non-adaptive halfplane queries and adaptive halfplane queries individually.
Same as before, for non-adaptive queries (Section~\ref{sec-nonadapthplane}), the algorithm is simple and the lower bound proof is nontrivial, while for adaptive queries (Section~\ref{sec-adapthplane}), the algorithm is more technical and the lower bound proof is simpler.

\subsection{Non-adaptive halfplane queries} \label{sec-nonadapthplane}
We first present our algorithmic result with non-adaptive halfplane emptiness queries.
Let $P \subseteq [0,1]^2$ be a (unknown) set of points.
We shall show how to approximate $\mathcal{CH}(P)$ with error $O(\frac{1}{q}+\delta)$ via $q$ non-adaptive queries to a $\delta$-accurate extreme halfplane oracle $\mathcal{E}$ on $P$.
Setting $\delta = \frac{1}{q}$ and applying Lemma~\ref{lem-extremetoemptiness} will give us the desired result.

\begin{algorithm}
    \caption{\textsc{NonAdaptiveHalfplane}$(q)$}
    \begin{algorithmic}[1]
        \State $V \leftarrow \{(\sin \frac{2i\pi}{q},\cos \frac{2i\pi}{q}) : i \in [q]\}$
        \State $H_{\vec{v}} \leftarrow \textsc{Query}(\vec{v})$ for all $\vec{v} \in V$
        \State $C^* \leftarrow \bigcap_{\vec{v} \in V} H_{\vec{v}}$
        \State \textbf{return} $C^* \cap [0,1]^2$
    \end{algorithmic}
    \label{alg-nonadahplane}
\end{algorithm}

For a unit vector $\vec{v} \in \mathbb{S}^1$, let $\textsc{Query}(\vec{v})$ denote the output of $\mathcal{E}$ when queried with $\vec{v}$.
Our algorithm for approximating $\mathcal{CH}(P)$ is very simple, and is presented in Algorithm~\ref{alg-nonadahplane}.
Let $V = \{(\sin \frac{2i\pi}{q},\cos \frac{2i\pi}{q}) : i \in [q]\}$ be $q$ unit vectors uniformly picked on $\mathbb{S}^1$.
We query $\mathcal{E}$ with the vectors in $V$, and let $H_{\vec{v}} = \textsc{Query}(\vec{v})$ for all $\vec{v} \in V$.
After that, we simply return $C^* \cap [0,1]^2$ where $C^* = \bigcap_{\vec{v} \in V} H_{\vec{v}}$.
Clearly, the number of queries to $\mathcal{E}$ is $q$, and they are non-adaptive as the algorithm performs them at the same time. 

We now bound the error to $O(1/\sqrt{q})$.
Write $\vec{v}_i = (\sin \frac{2i\pi}{q},\cos \frac{2i\pi}{q})$.
Then $V = \{\vec{v}_1,\dots,\vec{v}_q\}$, and let $\vec{v}_0 = \vec{v}_q$.
Define $L_{{\vec{v}_i}}$ as the minimal halfplane with normal vector $\vec{v}_i$ that contains $P$ and $C = \bigcap_{i=1}^q L_{{\vec{v}_i}}$.
We have $C \subseteq C^*$, since $L_{\vec{v}_i} \subseteq H_{\vec{v}_i}$ for all $i \in [n]$.
%Now $\mathcal{CH}(P) \subseteq C \subseteq C^* \cap [0,1]^2$.
%By Lemma~\ref{lem-enlarge}, $\lVert (C^* \cap [0,1]^2) \backslash C \rVert = O(\delta)$.
%So it suffices to show $\lVert C \backslash \mathcal{CH}(P) \rVert = O(\frac{1}{q^2}+\delta)$.
%Define $\mu_i = C \cap \partial L_{\vec{v}_i}$ for $i \in [n]$ and $\mu_0 = \mu_n$.
We first bound $\lVert C \backslash \mathcal{CH}(P) \rVert$.

\begin{lemma} \label{lem-CandCHP}
    $\lVert C \backslash \mathcal{CH}(P) \rVert = O(\frac{1}{q})$.
\end{lemma}
\begin{proof}
Define $\mu_i = C \cap \partial L_{\vec{v}_i}$ for $i \in [n]$ and $\mu_0 = \mu_n$.
The definition of each $L_{\vec{v}_i}$ guarantees that there exists a point $p_i \in P$ satisfying $p_i \in \partial L_{\vec{v}_i}$.
As $p_i \in P \subseteq C$, we have $p_i \in \mu_i$.
Let $P' = \{p_1,\dots,p_q\}$.
We have $P' \subseteq P$ and thus $\lVert C \backslash \mathcal{CH}(P') \rVert \geq \lVert C \backslash \mathcal{CH}(P) \rVert$.
It suffices to show $\lVert C \backslash \mathcal{CH}(P') \rVert = O(\frac{1}{q})$.
For each $i \in [q]$, let $a_i$ be the intersection point of $\mu_{i-1}$ and $\mu_i$, and $\tau_i$ be the segment with endpoints $p_{i-1}$ and $p_i$.
Note that $\lVert C \backslash \mathcal{CH}(P') \rVert = \sum_{i=1}^q \lVert \triangle p_{i-1} p_i a_i \rVert$.
Also, we have $\lVert \triangle p_{i-1} p_i a_i \rVert \leq |\tau_i| \cdot |\mu_i| \cdot \sin \angle a_ip_i p_{i-1} \leq |\tau_i| \cdot |\mu_i| \cdot \mathsf{ang}(\vec{v}_{i-1},\vec{v}_i)$.
Since $|\mu_i| = O(1)$ and $\mathsf{ang}(\vec{v}_{i-1},\vec{v}_i) = \frac{2\pi}{q}$, we have $\lVert \triangle p_{i-1} p_i a_i \rVert = O(\frac{|\tau_i|}{q})$.
It follows that $\sum_{i=1}^q \lVert \triangle p_{i-1} p_i a_i \rVert = O(\sum_{i=1}^q \frac{|\tau_i|}{q})$.
We have $\sum_{i=1}^n |\tau_i| = O(1)$, as $\tau_1,\dots,\tau_q$ are edges of the convex polygon $\mathcal{CH}(P')$ which is contained in $[0,1]^d$.
Therefore, $\lVert C \backslash \mathcal{CH}(P') \rVert = \sum_{i=1}^q \lVert \triangle p_{i-1} p_i a_i \rVert = O(\frac{1}{q})$.
\end{proof}

Set $C_0 = C\cap [0,1]^2$ and $C_0^* = C^* \cap [0,1]^2$.
Since $P \subseteq [0,1]^2$, the above lemma implies $\lVert C_0 \backslash \mathcal{CH}(P) \rVert = O(\frac{1}{q})$ as well.
Next, we show $\lVert C_0^* \backslash C_0 \rVert = O(\delta)$.
Let $\mathcal{L}$ (resp., $\mathcal{H}$) consist of the halfplanes $L_{\vec{v}_1},\dots,L_{\vec{v}_q}$ (resp., $H_{\vec{v}_1},\dots,H_{\vec{v}_q}$) and four additional halfplanes defined by the equations $x \geq 0$, $x \leq 1$, $y \geq 0$, and $y \leq 1$, respectively.
Clearly, $C_0 = \bigcap_{L \in \mathcal{L}} L$ and $C_0^* = \bigcap_{H \in \mathcal{H}} H$.
Since each $H_{\vec{v}_i}$ is a halfplane parallel to $L_{\vec{v}_i}$ satisfying $C_0 \subseteq H_{\vec{v}_i}$ and $\mathsf{dist}(\partial H_{\vec{v}_i}, \partial L_{\vec{v}_i}) \leq \delta$, Lemma~\ref{lem-enlarge} implies $\lVert (C_0^* \cap [0,1]^2) \backslash C_0 \rVert = O(\delta)$, i.e., $\lVert C_0^* \backslash C_0 \rVert = O(\delta)$.
Combining this with Lemma~\ref{lem-CandCHP}, we have $\lVert C_0^* \backslash \mathcal{CH}(P) \rVert = O(\frac{1}{q}+\delta)$, which is the error of Algorithm~\ref{alg-nonadahplane}.
Setting $\delta = \frac{1}{q}$, we see that one can approximate $\mathcal{CH}(P)$ with error $O(\frac{1}{q})$ via $q$ non-adaptive queries to a $\frac{1}{q}$-accurate extreme halfplane oracle.
Lemma~\ref{lem-extremetoemptiness} shows that a query to a $\frac{1}{q}$-accurate extreme halfplane oracle can be simulated with $O(q)$ non-adaptive halfplane emptiness queries.
Therefore, one can approximate $\mathcal{CH}(P)$ with error $O(\frac{1}{q})$ via $O(q^2)$ non-adaptive halfplane emptiness queries.
Equivalently, if we are allowed to perform $q$ queries, then the error achieved is $O(\frac{1}{\sqrt{q}})$.

\begin{theorem} \label{thm-nonadapt_hplane_ub}
    There exists an algorithm for \textnormal{\sc Approximate Convex Hull} in $\mathbb{R}^2$ that performs $q$ non-adaptive halfplane emptiness queries and has error $O(1/\sqrt{q})$.
\end{theorem}

Next, we show that any deterministic algorithm that performs $q$ non-adaptive halfplane queries must have error $\Omega(1/\sqrt{q})$.
Consider an algorithm $\mathbf{A}$.
As the $q$ queries performed by $\mathbf{A}$ are non-adaptive, they are independent of the point set $P$.
Let $H_1,\dots,H_q$ be these queries, each of which is a halfplane in $\mathbb{R}^2$.
For $x,y \in (0,1)$, denote by $\ell_{x,y}$ the line that goes through the points $(x,0)$ and $(0,y)$.
Also, for $X,Y \subseteq (0,1)$, define $\mathcal{L}_{X,Y} = \{\ell_{x,y}: x \in X \text{ and } y \in Y\}$.
%Also, define $p_{x,y,t} = (\frac{x}{2}+t,\frac{y}{2}+t)$ for $t \in \mathbb{R}$.
%That is, $p_{x,y,t}$ is the point with offset $(t,t)$ from the midpoint of $\ell_{x,y} \cap [0,1]^2$, $(\frac{x}{2},\frac{y}{2})$.
Fix a parameter $\delta = \frac{1}{10\sqrt{q}}$, and let $A = \{i\delta: i \in \mathbb{Z}\} \cap [\frac{1}{3},\frac{2}{3}]$.
%For each $(\alpha,\beta)\in I\times I$, let $\sigma_{\alpha,\beta}$ be the segment in $[0,1]^2$ connecting the points $(\alpha,0)$ and $(0,\beta)$.
%Also, define $p_{\alpha,\beta} = (\frac{\alpha}{2}+\frac{\delta}{10},\frac{\beta}{2}+\frac{\delta}{10})$

\begin{lemma} \label{lem-emptybucket}
    There exists $a,b \in A$ such that $\partial H_i \notin \mathcal{L}_{[a,a+\delta),[b,b+\delta)}$ for all $i \in [q]$.
\end{lemma}
\begin{proof}
Observe that for different $a,a' \in A$, the intervals $[a,a+\delta)$ and $[a',a'+\delta)$ are disjoint.
Therefore, $\mathcal{L}_{[a,a+\delta),[b,b+\delta)}$ and $\mathcal{L}_{[a',a'+\delta),[b',b'+\delta)}$ are disjoint, for $(a,b),(a',b') \in A \times A$ with $(a,b) \neq (a',b')$.
It follows that for each $i \in [q]$, there exists at most one pair $(a,b) \in A \times A$ such that $\partial H_i \in \mathcal{L}_{[a,a+\delta),[b,b+\delta)}$; we charge $i$ to $(a,b)$ if this is the case.
By construction, $|A| \geq 3\sqrt{q}$ and thus $|A \times A| \geq 9q > q$.
As such, there exists a pair $(a,b) \in A \times A$ such that no number in $[q]$ is charged to $(a,b)$.
We then have $\partial H_i \notin \mathcal{L}_{[a,a+\delta),[b,b+\delta)}$ for all $i \in [q]$.
\end{proof}

Let $a,b \in A$ be as in the above lemma, and $p^* = (\frac{a}{2}+\frac{\delta}{10},\frac{b}{2}+\frac{\delta}{10})$.
We define two sets $Z = \{(a,0),(0,b)\}$ and $Z' = \{(a,0),(0,b),p^*\}$; see figure \ref{fig:nonadapt_lb}.

\begin{lemma} \label{lem-below}
    The point $p^*$ is below the lines $\ell_{a+\delta,b}$ and $\ell_{a,b+\delta}$.
\end{lemma}
\begin{proof}
Without loss of generality, we only need to show $p^*$ is below $\ell_{a+\delta,b}$.
The equation of $\ell_{a+\delta,b}$ is $bx+(a+\delta)y - b(a+\delta) = 0$.
Since $a,b \in A$, we have $a,b \in [\frac{1}{3},\frac{2}{3}]$.
So we have
\begin{align*}
    b\left(\frac{a}{2}+\frac{\delta}{10}\right) +(a+\delta) \left(\frac{b}{2}+\frac{\delta}{10}\right) - b(a+\delta) = \left(\frac{a+b +\delta}{10} - \frac{b}{2}\right) \delta < 0,
\end{align*}
which implies that that $p^*$ is below $\ell_{a+\delta,b}$.
\end{proof}

\begin{corollary}
    For all $i \in [q]$, $Z \cap H_i = \emptyset$ iff $Z' \cap H_i = \emptyset$.
\end{corollary}
\begin{proof}
If $p^* \notin H_i$, then $Z \cap H_i = Z' \cap H_i$ and we are done.
So suppose $p^* \in H_i$.
We claim that either $(a,0) \in H_i$ or $(0,b) \in H_i$.
It suffices to show that at least one of $(a,0)$ and $(0,b)$ lies on the same side of $\partial H_i$ as $p^*$.
Assume both $(a,0)$ and $(0,b)$ are on the opposite side of $\partial H_i$ against $p^*$.
Then $p^*$ and the segment $\ell_{a,b} \cap [0,1]^2$ must lie on the opposite side of $\partial H_i$.
But by our construction, $p^*$ is sufficiently close to $\ell_{a,b} \cap [0,1]^2$, which forces $\partial H_i = \ell_{x,y}$ for some $x \in [a,1]$ and $y \in [b,1]$.
By Lemma~\ref{lem-emptybucket}, we cannot have $x \in [a,a+\delta)$ and $y \in [b,b+\delta)$ at the same time.
So either $x \geq a+\delta$ or $y \geq b+\delta$.
Without loss of generality, assume $x \geq a+\delta$.
By Lemma~\ref{lem-below}, $p^*$ is below $\ell_{a+\delta,b}$.
Since $x \geq a+\delta$ and $y \geq b$, $p^*$ is below $\partial H_i$ as well.
But both $(a,0)$ and $(0,b)$ are also below $\partial H_i$, contradicting our assumption.
As such, we have either $(a,0) \in H_i$ or $(0,b) \in H_i$.
It follows that $Z \cap H_i \neq \emptyset$ and $Z' \cap H_i \neq \emptyset$, which completes the proof.
\end{proof}

The rest of the proof is almost the same as that of Theorem~\ref{thm-nonadapt_orth_lb}.
The above lemma shows that the algorithm $\mathbf{A}$ cannot distinguish $Z$ and $Z'$, that is, it returns the same convex polygon when running on $P = Z$ and $P = Z'$.
We only need to observe the following simple fact.
\begin{lemma}
\label{lem-totalvol'}
    $\lVert \mathcal{CH}(Z') \backslash \mathcal{CH}(Z) \rVert = \Omega(\delta)$.
\end{lemma}
\begin{proof}
Clearly, $\lVert \mathcal{CH}(Z) \rVert = 0$, it suffices to show $\lVert \mathcal{CH}(Z')\rVert = \Omega(\delta)$.
The distance from $p$ to $\ell_{a,b}$ is $\Omega(\delta)$ by our construction.
Thus, $\lVert \mathcal{CH}(Z')\rVert = \Omega(\delta \cdot \sqrt{a^2+b^2})$.
We have $\sqrt{a^2+b^2} = \Omega(1)$, because $a,b \in [\frac{1}{3},\frac{2}{3}]$.
It follows that $\lVert \mathcal{CH}(Z')\rVert = \Omega(\delta)$.
\end{proof}

Let $C^*$ be the output of $\mathbf{A}$ when running on $P = Z$ or $P = Z'$.
We have $\lVert \mathcal{CH}(Z') \backslash \mathcal{CH}(Z) \rVert \leq \lVert C^* \oplus \mathcal{CH}(Z) \rVert + \lVert C^* \oplus \mathcal{CH}(Z') \rVert$.
Thus, the above lemma implies either $\lVert C^* \oplus \mathcal{CH}(Z) \rVert = \Omega(\delta)$ or $\lVert C^* \oplus \mathcal{CH}(Z') \rVert = \Omega(\delta)$.
As $\delta = \Theta(1/\sqrt{q})$, the error of $\mathbf{A}$ is $\Omega(1/\sqrt{q})$.

\begin{figure}
\centering
\begin{tikzpicture}[scale=7,>=stealth]

% PARAMETERS -----------------------------------------------------------
\def\c{0.15}
\def\eps{\t* \c/1.5}

% Precompute tick positions (4 ticks from c to 1-c)
\pgfmathsetmacro{\t}{(1-2*\c)/4}  % spacing between ticks
\pgfmathsetmacro{\a}{\c + \t}
\pgfmathsetmacro{\b}{\c + 2*\t}

% Derived points
\coordinate (A)  at (\a,0);
\coordinate (Ap) at (\a+\t,0);
\coordinate (B) at (0,\b);
\coordinate (Bp) at (0,\b+\t);
\coordinate (Mid) at ({\a/2},{\b/2});
\coordinate (Pab) at ({\a/2 + \eps},{\b/2 + \eps});

% UNIT SQUARE ---------------------------------------------------------
\draw[thick] (0,0) rectangle (1,1);

\node[left]  at (0,1) {$(0,1)$};
\node[below] at (1,0) {$(1,0)$};

% TICKS ---------------------------------------------------------------
% x-axis
\foreach \i in {0,1,2,3,4}{
    \draw (\c + \i*\t,0.01) -- (\c + \i*\t,-0.01);
}
% y-axis
\foreach \i in {0,1,2,3,4}{
    \draw (0.01,\c + \i*\t) -- (-0.01,\c + \i*\t);
}

% Labels
\node[below] at (\a,0) {$x$};
\node[below] at (\a+\t,0) {$x+\delta$};
\node[left]  at (0,\b) {$y$};
\node[left]  at (0,\b+\t) {$y+\delta$};

\node[below] at (\c,0) {$1/3$};
\node[below] at (1-\c,0) {$2/3$};
\node[left] at (0,\c) {$1/3$};
\node[left] at (0,1-\c) {$2/3$};

% t spacers
%\node at ({\a+\t/2}, -0.07) {$t$};
%\node at (-0.07, {\b+\t/2}) {$t$};

% GRAY CONNECTIONS (FAINT) --------------------------------------------
\foreach \i in {0,1,2,3,4}{
    \foreach \j in {0,1,2,3,4}{
        \draw[very thin,gray!40] (\c + \i*\t,0) -- (0,\c + \j*\t);
    }
}

% DARK IMPORTANT SEGMENTS ---------------------------------------------
\draw[ultra thick, blue] (A) -- (B);
\draw[thick] (Ap) -- (B);
\draw[thick] (A) -- (Bp);

% KEY POINTS -----------------------------------------------------------
\fill (A) circle (0.015);
\fill (B) circle (0.015);
\fill (Ap) circle (0.015);
\fill (Bp) circle (0.015);

\fill (Mid) circle (0.012);
\node[below, xshift=-10pt] at (Mid) {$\left(\frac{x}{2},\frac{y}{2}\right)$};

\fill[red] (Pab) circle (0.012);
\node[right, yshift=5pt] at (Pab) {$p^* = (\frac{x}{2} + \frac{\delta}{10}, \frac{y}{2} + \frac{\delta}{10})$};

\end{tikzpicture}
\caption{Lower bound construction for non-adaptive halfplane queries. The point $p^*$ (red) is offset from the midpoint $(x/2, y/2)$ of $l_{x,y}$ to lie sufficiently far above $l_{x,y}$ (blue) but below both $l_{x+\delta,y}$ and $l_{x,y+\delta}$. Only a segment with both $x$-intercept in the interval $[x, x+\delta]$ and $y$-intercept in the interval $[y, y+\delta]$ can separate both $(x, 0), (0, y)$ from $p^*$.
Doing so for all $\Theta(1/\delta^2) > q$ paired choices of $x, y$ (gray lines) necessitates as many queries.
Range [1/3, 2/3] enlarged for clarity.}
\label{fig:nonadapt_lb}
\end{figure}
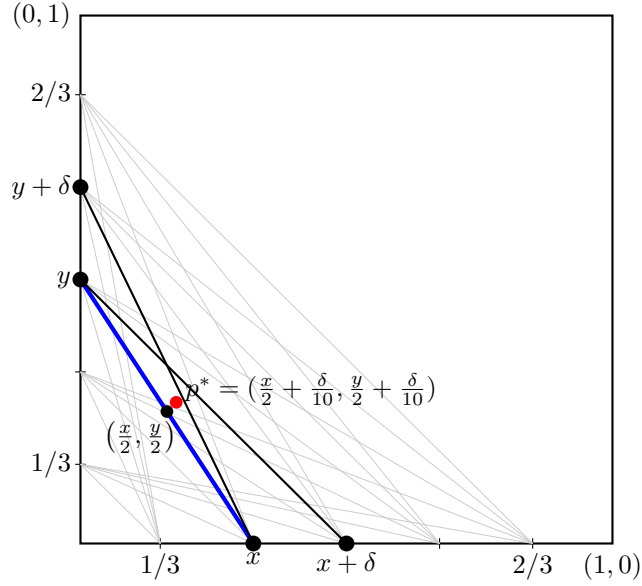

\begin{theorem} \label{thm-nonadapt_hplane_lb}
    Any deterministic algorithm for \textnormal{\sc Approximate Convex Hull} in $\mathbb{R}^2$ that performs $q$ non-adaptive halfplane emptiness queries incurs error $\Omega(1/\sqrt{q})$.
\end{theorem}

\subsection{Adaptive halfplane queries} \label{sec-adapthplane}

We first present our algorithm with adaptive halfplane emptiness queries.
Let $P \subseteq [0,1]^2$ be a (unknown) set of points.
We shall show how to approximate $\mathcal{CH}(P)$ with error $O(\frac{1}{q^2} + q^2\delta)$ via $O(q \log q)$ adaptive queries to a $\delta$-accurate extreme halfplane oracle $\mathcal{E}$ on $P$.
Combining this with Lemma~\ref{lem-extremetoemptiness} will give us the desired result.

Same as before, for a unit vector $\vec{v} \in \mathbb{S}^1$, let $\textsc{Query}(\vec{v})$ denote the output of $\mathcal{E}$ when queried with $\vec{v}$.
%An extreme halfplane query $\textsc{Query}(\vec{v})$ takes a unit vector $\vec{v} \in \mathbb{S}^1$ as input and returns the minimal halfplane $H$ with normal vector $\vec{v}$ such that $P \subseteq H$.
Our algorithm for approximating $\mathcal{CH}(P)$ is presented in Algorithm~\ref{alg-adahplane}.
We begin with a set $V$ of unit vectors which is initially $\{(1,0),(0,1),(-1,0),(0,-1)\}$.
We query $\mathcal{E}$ with the vectors in $V$ and let $H_{\vec{v}} = \textsc{Query}(\vec{v})$.
Then we iteratively add more vectors to $V$ and query $\mathcal{E}$ with these vectors (the repeat-until loop in the algorithm).
In each round, we first compute $C^* = \bigcap_{\vec{v} \in V} H_{\vec{v}}$ for the current $V$ (line~4).
Then line~5 sorts the vectors in the current $V$ in clockwise order, and let $\vec{v}_1,\dots,\vec{v}_n$ be the sorted sequence.
Let $\vec{v}_0 = \vec{v}_n$ and $\sigma_i = C^* \cap \partial H_{\vec{v}_i}$, which is the edge of $C^*$ corresponding to the halfplane $H_{\vec{v}_i}$.
We consider each pair $(\vec{v}_{i-1},\vec{v}_i)$ of adjacent vectors satisfying $\mathsf{ang}(\vec{v}_{i-1},\vec{v}_i) \cdot |\sigma_i| > \frac{1}{q^2}$, where $\mathsf{ang}(\vec{v}_{i-1},\vec{v}_i)$ denotes the magnitude of the angle from $\vec{v}_{i-1}$ to $\vec{v}_i$.
%Let $\mathsf{ang}(\vec{v}_{i-1},\vec{v}_i)$ denote the magnitude of the angle from $\vec{v}_{i-1}$ to $\vec{v}_i$.
%If $\mathsf{ang}(\vec{v}_{i-1},\vec{v}_i) \cdot |\sigma_i| > \frac{1}{q^2}$, 
For each such pair $(\vec{v}_{i-1},\vec{v}_i)$, we add a new unit vector $\vec{v}$ to $V$, which is the bisector of $\vec{v}_{i-1}$ and $\vec{v}_i$, and query $\mathcal{E}$ with $\vec{v}$ to obtain the halfplane $H_{\vec{v}} = \textsc{Query}(\vec{v})$.
After this, we proceed to the next round.
The procedure terminates when no more vectors are added to $V$ in a round.
At the end, the algorithm returns the convex polygon $C^* \cap [0,1]^2$ as its output.

\begin{algorithm}
    \caption{\textsc{AdaptiveHalfplane}$(q)$}
    \begin{algorithmic}[1]
        \State $V \leftarrow \{(1,0),(0,1),(-1,0),(0,-1)\}$
        \State $H_{\vec{v}} \leftarrow \textsc{Query}(\vec{v})$ for all $\vec{v} \in V$
        \Repeat
            \State $C^* \leftarrow \bigcap_{\vec{v} \in V} H_{\vec{v}}$        
            \State $(\vec{v}_1,\dots,\vec{v}_n) \leftarrow \textsc{Sort}(V)$
            \State $\vec{v}_0 \leftarrow \vec{v}_n$
            %\State $H_i \leftarrow \textsc{Query}(\vec{v}_i)$ for $i \in [n]_0$
            \State $\sigma_i \leftarrow C^* \cap \partial H_{\vec{v}_i}$ for $i \in [n]_0$
            \State $I \leftarrow \{i \in [n]: \mathsf{ang}(\vec{v}_{i-1},\vec{v}_i) \cdot |\sigma_i| > \frac{1}{q^2} \}$
            \For{every $i \in I$}
                \State $\vec{v} \leftarrow (\vec{v}_{i-1}+\vec{v}_i)/ \lVert \vec{v}_{i-1}+\vec{v}_i \rVert_2$
                \State $V \leftarrow V \cup \{\vec{v}\}$
                \State $H_{\vec{v}} \leftarrow \textsc{Query}(\vec{v})$
            \EndFor
        \Until{$I = \emptyset$}
        \State \textbf{return} $C^* \cap [0,1]^2$
    \end{algorithmic}
    \label{alg-adahplane}
\end{algorithm}

\subparagraph{Bounding the number of queries.}
Clearly, the number of queries we made in Algorithm~\ref{alg-adahplane} is just equal to the number of vectors in $V$ at the end of the algorithm.
Suppose the repeat-until loop in Algorithm~\ref{alg-adahplane} has $t$ iterations.
For $i \in [t]$, let $V_i$ be the set $V$ at the beginning of the $i$-th iteration and $C_i^* = \bigcap_{\vec{v} \in V_i} H_{\vec{v}}$ which is just the convex polygon $C^*$ in the $i$-th iteration.
Note that $V_t$ is just the set $V$ at the end of the algorithm, as the algorithm does not add new vectors to $V$ in the last iteration.
%Denote by $A_i$ the set of angles between adjacent vectors in $V_i$.
For convenience, we write $V_0 = \emptyset$.
%we write $V_{t+1}$ and $C_{t+1}^*$ as the set $V$ and the convex polygon $C^*$ at the end of the algorithm, respectively.
%Also, set $V_0 = \emptyset$.
Denote by $\mathcal{A}(V_i)$ the set of angles between adjacent vectors in $V_i$.
We first bound $|V_i \backslash V_{i-1}|$ for $i \in [t]$.
%The proofs in this subsection are included in the appendix.
%and then bound the number $t$ of iterations.

\begin{lemma}
\label{lem-23}
    $|V_i \backslash V_{i-1}| = O(q)$ for all $i \in [t]$.
\end{lemma}
\begin{proof}
Clearly, we have $|V_i \backslash V_{i-1}| = |\mathcal{A}(V_{i-1}) \backslash \mathcal{A}(V_i)|$.
Indeed, the vectors in $V_i \backslash V_{i-1}$ are those added to $V$ in the $(i-1)$-th iteration by the algorithm, which one-to-one corresponds to the angles in $\mathcal{A}(V_{i-1})$ that split into two in $\mathcal{A}(V_i)$.
Suppose $\mathcal{A}(V_{i-1}) \backslash \mathcal{A}(V_i) = \{\alpha_1,\dots,\alpha_k\}$.
It suffices to show $k = O(q)$.
For $j \in [k]$, let $\vec{v}_j \in V_{i-1}$ be the vector incident to $\alpha_j$ that is the clockwise boundary of $\alpha_j$.
Define $\tau_j = C_{i-1}^* \cap \partial H_{\vec{v}_j}$.
As the algorithm adds to $V$ a vector in the angle $\alpha_j$ in the $(i-1)$-th iteration, we know that $|\alpha_j| \cdot |\tau_j| > \frac{1}{q^2}$, which implies either $|\alpha_j| > \frac{1}{q}$ or $|\tau_j| > \frac{1}{q}$.
Let $J = \{j \in [k]: |\alpha_j| > \frac{1}{q}\}$ and thus $[k] \backslash J = \{j \in k: |\tau_j| > \frac{1}{q}\}$.
We have $|J| \leq 2\pi q = O(q)$, since $\sum_{j \in J} |\alpha_j| \leq \sum_{j=1}^k |\alpha_j| \leq 2 \pi$.
Next, we observe $k-|J| = O(q)$.
The vectors $\vec{v}_1,\dots,\vec{v}_k$ are distinct, as they are the clockwise boundaries of $\alpha_1,\dots,\alpha_k$, respectively.
Therefore, $\sigma_1,\dots,\sigma_k$ are distinct edges of $C_{i-1}^*$.
It follows that $\sum_{j=1}^k |\tau_j|$ is bounded by the perimeter of $C_{i-1}^*$, which is $O(1)$.
We then have $\sum_{j \in [k] \backslash J} |\tau_j| = O(1)$.
But $|\tau_j|> \frac{1}{q}$ for all $j \in [k] \backslash J$, which implies $k-|J| = |[k] \backslash J| = O(q)$.
Thus, $k = |J| + (k-|J|) = O(q)$.
\end{proof}

The above lemma directly implies $|V_t| = O(qt)$.
Thus, it suffices to bound $t$.
To this end, we establish the following properties for the angles in each $\mathcal{A}(V_i)$.

\begin{lemma} \label{lem-1/q^2}
    For all $i \in [t]$ and $\alpha \in \mathcal{A}(V_i)$, we have $|\alpha| \geq \frac{1}{12q^2}$.
\end{lemma}
\begin{proof}
We apply induction on $i$.
Clearly, $|\alpha| = \frac{\pi}{2} \geq \frac{1}{12q^2}$ for $\alpha \in \mathcal{A}(V_1)$.
Suppose $|\alpha| \geq \frac{1}{12q^2}$ for all $\alpha \in \mathcal{A}(V_{i-1})$.
Consider an angle $\alpha \in \mathcal{A}(V_i)$.
If $\alpha \in \mathcal{A}(V_{i-1})$, we are done.
Otherwise, $\alpha$ is contained in a unique angle $\alpha' \in \mathcal{A}(V_{i-1})$ with $|\alpha'| = 2|\alpha|$.
Let $\vec{v} \in V_{i-1}$ be the vector incident to $\alpha'$ that is the clockwise boundary of $\alpha'$.
Define $\sigma = C_{i-1}^* \cap \partial H_{\vec{v}}$.
We have $|\sigma| \leq 6$, as $\sigma \subseteq [-\delta,1+\delta]^2$.
Because the angle $\alpha'$ is split into two (one of which is $\alpha$) during the $(i-1)$-th iteration, we have $\mathsf{ang}(\vec{v},\vec{v}') \cdot |\sigma| > \frac{1}{q^2}$, which implies $|\alpha'| = \mathsf{ang}(\vec{v},\vec{v}') > \frac{1}{6q^2}$.
As such, $|\alpha| \geq \frac{1}{12q^2}$.
\end{proof}

\begin{lemma} \label{lem-survive}
    Let $\alpha \in \mathcal{A}(V_i)$.
    If $\alpha \in \mathcal{A}(V_{i+1})$, then $\alpha \in \mathcal{A}(V_j)$ for all $j \in \{i,\dots,t\}$.
\end{lemma}
\begin{proof}
Let $\alpha \in \mathcal{A}(V_i) \cap \mathcal{A}(V_{i+1})$.
It suffices to show $\alpha \in \mathcal{A}(V_{i+2})$, which would implies $\alpha \in \mathcal{A}(V_j)$ for all $j \geq i+2$ by an simple induction argument.
%Suppose $\alpha$ is the angle between $\vec{v},\vec{v}' \in V_i$.
Let $\vec{v} \in V_i$ be the vector incident to $\alpha$ that is the clockwise boundary of $\alpha$ and $\sigma = C_i^* \cap H_{\vec{v}}$.
As $\alpha$ survives in $\mathcal{A}(V_{i+1})$, we have $\mathsf{ang}(\vec{v},\vec{v}') \cdot |\sigma| \leq \frac{1}{q^2}$.
Now notices that $C_{i+1}^* \cap H_{\vec{v}} \subseteq \sigma$, simply because $C_{i+1}^* \subseteq C_i^*$.
This directly implies $\alpha \in \mathcal{A}(V_{i+2})$.
\end{proof}

\begin{lemma} \label{lem-2^i}
    Let $i \in [t]$.
    We have $|\alpha| = \frac{\pi}{2^i}$ for all $\alpha \in \mathcal{A}(V_i) \backslash \mathcal{A}(V_{i+1})$.
    %and $|\alpha| \geq \frac{1}{2q^2}$ for all $\alpha \in \mathcal{A}(V_{t+1})$.
\end{lemma}
\begin{proof}
%To show $|\alpha| = \frac{\pi}{2^i}$ for all $\alpha \in \mathcal{A}(V_i) \backslash \mathcal{A}(V_{i+1})$, 
We apply induction on $i$.
When $i = 1$, the equation holds since $|\alpha| = \frac{\pi}{2}$ for all $\alpha \in \mathcal{A}(V_1)$.
Assume $|\alpha| = \frac{\pi}{2^{i-1}}$ for all $\alpha \in \mathcal{A}(V_{i-1}) \backslash \mathcal{A}(V_i)$.
Let $\alpha \in \mathcal{A}(V_i) \backslash \mathcal{A}(V_{i+1})$.
There is a unique angle $\alpha' \in \mathcal{A}(V_{i-1})$ that contains $\alpha$.
As $\alpha \in \mathcal{A}(V_i)$ and $\alpha \notin \mathcal{A}(V_{i+1})$, by Lemma~\ref{lem-survive}, we have $\alpha \notin \mathcal{A}(V_{i-1})$.
Thus, $\alpha \neq \alpha'$ and $|\alpha| = |\alpha'|/2$.
Since $\alpha'$ is split into two smaller angles in $\mathcal{A}(V_i)$ (one of which is $\alpha$), we have $\alpha' \notin \mathcal{A}(V_i)$ and thus $\alpha' \in \mathcal{A}(V_{i-1}) \backslash \mathcal{A}(V_i)$.
By our induction hypothesis, $|\alpha'| = \frac{\pi}{2^{i-1}}$, which implies $|\alpha| = \frac{\pi}{2^i}$.
\end{proof}
Lemma~\ref{lem-1/q^2} and Lemma~\ref{lem-2^i} together imply that $t = O(\log q)$.
Indeed, $\mathcal{A}(V_{i}) \neq \mathcal{A}(V_{i+1})$ for all $i \in [t-1]$, for otherwise the algorithm terminates after the $i$-th iteration.
Thus there exists $\alpha \in \mathcal{A}(V_{t-1}) \backslash \mathcal{A}(V_t)$.
We have $\frac{\pi}{2^{t-1}} = |\alpha| \geq \frac{1}{q^2}$ by Lemma~\ref{lem-1/q^2} and Lemma~\ref{lem-2^i}, which implies $t = O(\log q)$.
It follows that $|V_t| = O(qt) = O(q \log q)$.

\subparagraph{Bounding the error.}
Next, we show that the output $C^* \cap [0,1]^2$ of Algorithm~\ref{alg-adahplane} satisfies $\lVert (C^* \cap [0,1]^2) \oplus \mathcal{CH}(P) \rVert = O(\frac{1}{q^2}+\delta)$.
Let $V = \{\vec{v}_1,\dots,\vec{v}_n\}$ be the set of vectors at the end of Algorithm~\ref{alg-adahplane}, where $\vec{v}_1,\dots,\vec{v}_n$ are sorted in clockwise order.
Set $\vec{v}_0 = \vec{v}_n$.
As in the algorithm, for each $i \in [n]_0$, let $H_{\vec{v}_i} = \textsc{Query}(\vec{v}_i)$ and $\sigma_i = C^* \cap \partial H_{\vec{v}_i}$.
Then $C^* = \bigcap_{i=1}^n H_{\vec{v}_i}$.
We have $|\mathsf{ang}(\vec{v}_{i-1},\vec{v}_i)| \cdot (|\sigma_{i-1}|+|\sigma_i|) > \frac{1}{q^2}$ for all $i \in [n]_0$, for otherwise the repeat-until loop in Algorithm~\ref{alg-adahplane} cannot terminate.
Define $L_{{\vec{v}_i}}$ as the minimal halfplane with normal vector $\vec{v}_i$ that contains $P$ and $C = \bigcap_{i=1}^n L_{{\vec{v}_i}}$.
We have $C \subseteq C^*$, since $L_{\vec{v}_i} \subseteq H_{\vec{v}_i}$ for all $i \in [n]$.
Also, we have $C \subseteq [0,1]^2$, since $V$ consists of the four vectors $(1,0),(0,1),(-1,0),(0,-1)$, which guarantees that $C$ is contained in the axis-parallel bounding box of $P$ and thus contained in $[0,1]^2$.
Now $\mathcal{CH}(P) \subseteq C \subseteq C^* \cap [0,1]^2$.
By Lemma~\ref{lem-enlarge}, $\lVert (C^* \cap [0,1]^2) \backslash C \rVert = O(\delta)$.
So it suffices to show $\lVert C \backslash \mathcal{CH}(P) \rVert = O(\frac{1}{q^2}+\delta)$.
Define $\mu_i = C \cap \partial L_{\vec{v}_i}$ for $i \in [n]$ and $\mu_0 = \mu_n$.

\begin{lemma} \label{lem-mu=sigma+}
    $|\mu_i| \leq |\sigma_i| + O(q^2\delta)$ for all $i \in [n]$.
    %$|\mu_i| \leq |\sigma_i| + O(\delta /\sin \alpha_i)$, where $\alpha_i = \min\{\mathsf{ang}(\vec{v}_{i-1},\vec{v}_i),\mathsf{ang}(\vec{v}_i,\vec{v}_{i+1})\}$.
\end{lemma}
\begin{proof}
Let $a$ and $a'$ be the endpoints of $\mu_i$, where $a$ (resp., $a'$) is on the counterclockwise (resp., clockwise) side of $\mu_i$ with respect to $C$.
We first consider the general case $\sigma_i \neq \emptyset$ and then discuss the special case $\sigma_i = \emptyset$.
If $\sigma_i \neq \emptyset$, let $b$ and $b'$ be the endpoints of $\sigma_i$, where $b$ (resp., $b'$) is on the counterclockwise (resp., clockwise) side of $\sigma_i$ with respect to $C^*$.
As $\sigma_i \subseteq \partial H_{\vec{v}_i}$ and $\mu_i \subseteq \partial L_{\vec{v}_i}$, the $4$-gon $aa'b'b$ is a trapezoid with parallel edges $\mu_i$ and $\sigma_i$.
Therefore, we have the equation
\begin{equation*}
    |\mu_i| = |\sigma_i| + \frac{h}{\tan \angle baa'} + \frac{h}{\tan \angle b'a'a}.
\end{equation*}
where $h$ is the height of this trapezoid.
We have $h \leq \delta$ since $\mathsf{dist}(\partial H_{\vec{v}_i},\partial L_{\vec{v}_i}) \leq \delta$.
Furthermore, we observe that $|\angle baa'| \geq \mathsf{ang}(\vec{v}_{i-1},\vec{v}_i)$ and $|\angle b'a'a| \geq \mathsf{ang}(\vec{v}_i,\vec{v}_{i+1})$.
Indeed, $\angle baa' = \pi - \angle b'ba$.
Since $a \in C^*$, the angle $\angle b'ba$ is smaller than or equal to the angle of $C^*$ at the vertex $b$, where the latter is just $\pi - \mathsf{ang}(\vec{v}_{i-1},\vec{v}_i)$.
Thus, 
\begin{equation*}
    |\angle baa'| = \pi - |\angle b'ba| \geq \pi -(\pi - \mathsf{ang}(\vec{v}_{i-1},\vec{v}_i)) = \mathsf{ang}(\vec{v}_{i-1},\vec{v}_i).
\end{equation*}
For the same reason, $|\angle b'a'a| \geq \mathsf{ang}(\vec{v}_i,\vec{v}_{i+1})$.
By Lemma~\ref{lem-1/q^2}, it follows that $|\angle baa'| \geq \frac{1}{12q^2}$ and $|\angle b'a'a| \geq \frac{1}{12q^2}$.
As such, $|\mu_i| = |\sigma_i| + O(q^2 \delta)$.

If $\sigma_i = \emptyset$, then $C^*$ is contained in the interior of $H_{\vec{v}_i}$.
In this case, $C^*$ has a unique vertex that is closest to $\partial H_{\vec{v}_i}$.
Let both $b$ and $b'$ be this vertex.
The same argument applies.
\end{proof}

The definition of each $L_{\vec{v}_i}$ guarantees that there exists a point $p_i \in P$ satisfying $p_i \in \partial L_{\vec{v}_i}$.
As $p_i \in P \subseteq C$, we have $p_i \in \mu_i$.
Let $P' = \{p_1,\dots,p_n\}$.
We have $P' \subseteq P$ and thus $\lVert C \backslash \mathcal{CH}(P') \rVert \geq \lVert C \backslash \mathcal{CH}(P) \rVert$.
We shall show $\lVert C \backslash \mathcal{CH}(P') \rVert = O(\frac{1}{q^2}+\delta)$, which bounds $\lVert C \backslash \mathcal{CH}(P) \rVert$ as well.
For each $i \in [n]$, let $a_i$ be the intersection point of $\mu_{i-1}$ and $\mu_i$, and $\tau_i$ be the segment with endpoints $p_{i-1}$ and $p_i$.
Note that $\lVert C \backslash \mathcal{CH}(P') \rVert = \sum_{i=1}^n \lVert \triangle p_{i-1} p_i a_i \rVert$.

\begin{lemma}
    $\lVert \triangle p_{i-1} p_i a_i \rVert = O(|\tau_i| \cdot (\frac{1}{q^2}+q^2\delta))$ for all $i \in [n]$.
\end{lemma}
\begin{proof}
By construction, the distance between $p_{i-1}$ and $p_i$ is just $|\tau_i|$.
So we only need to show that the height $h$ of $\triangle p_{i-1} p_i a_i$ with respect to the edge $p_{i-1} p_i$ is bounded by $O(\frac{1}{q^2}+q^2\delta)$.
Let $\theta = \angle a_i p_i p_{i-1}$.
As the distance between $a_i$ and $p_i$ is at most $|\mu_i|$, we have 
\begin{equation*}
    h \leq |\mu_i| \cdot \sin|\theta| \leq |\mu_i| \cdot \theta \leq \mathsf{ang}(\vec{v}_{i-1},\vec{v}_i) \cdot |\mu_i|.
\end{equation*}
By Lemma~\ref{lem-mu=sigma+}, this implies $h = \mathsf{ang}(\vec{v}_{i-1},\vec{v}_i) \cdot |\sigma_i| + O(q^2 \delta)$.
Since $\mathsf{ang}(\vec{v}_{i-1},\vec{v}_i) \cdot |\sigma_i| \leq \frac{1}{q^2}$, we have $h = O(\frac{1}{q^2}+q^2\delta)$ and thus $\lVert \triangle p_{i-1} p_i a_i \rVert = O(|\tau_i| \cdot (\frac{1}{q^2}+q^2\delta))$.
\end{proof}

The above lemma directly implies 
\begin{equation*}
    \lVert C \backslash \mathcal{CH}(P') \rVert = \sum_{i=1}^n \lVert \triangle p_{i-1} p_i a_i \rVert = O((\sum_{i=1}^n |\tau_i|) \cdot (\frac{1}{q^2}+q^2\delta)).
\end{equation*}
As $\tau_1,\dots,\tau_n$ are just the edge of the convex polygon $\mathcal{CH}(P')$, we have $\sum_{i=1}^n |\tau_i| = O(1)$ and thus $\lVert C \backslash \mathcal{CH}(P') \rVert = O(\frac{1}{q^2}+q^2\delta)$.
As Lemma~\ref{lem-enlarge} implies $\lVert (C^* \cap [0,1]^2) \backslash C \rVert = O(\delta)$, we have $\lVert (C^* \cap [0,1]^2) \oplus \mathcal{CH}(P) \rVert \leq \lVert (C^* \cap [0,1]^2) \backslash \mathcal{CH}(P') \rVert = O(\frac{1}{q^2}+q^2\delta)$.
Setting $\delta = \frac{1}{q^4}$, we see that one can approximate $\mathcal{CH}(P)$ with error $O(\frac{1}{q^2})$ via $O(q \log q)$ queries to a $\frac{1}{q^4}$-accurate extreme halfplane oracle for $P$.

Lemma~\ref{lem-extremetoemptiness} shows that a query to a $\frac{1}{q^4}$-accurate extreme halfplane oracle can be simulated with $O(\log q)$ adaptive halfplane emptiness queries.
Therefore, one can approximate $\mathcal{CH}(P)$ with error $O(\frac{1}{q^2})$ via $O(q \log^2 q)$ adaptive halfplane emptiness queries.
Equivalently, if we are allowed to use $q$ queries, then the error achieved can be bounded by $O(\frac{\log^2q}{q^2})$, i.e., $\widetilde{O}(\frac{1}{q^2})$.

\begin{theorem} \label{thm-adapt_hplane_ub}
    There exists an algorithm for \textnormal{\sc Approximate Convex Hull} in $\mathbb{R}^2$ that performs $q$ adaptive halfplane emptiness queries and has error $\widetilde{O}(1/q^2)$.
    %Let $P$ be an unknown set of points in $[0,1]^2$.
    %One can approximate $\mathcal{CH}(P)$ with error $\widetilde{O}(\frac{1}{q^2})$ by applying $q$ adaptive halfplane range-emptiness queries on $P$.
\end{theorem}

We remark that replacing the threshold $\frac{1}{q^2}$ in Algorithm~\ref{alg-adahplane} with an even smaller number cannot result in an improvement for Theorem~\ref{thm-adapt_hplane_ub}: while it can decrease the error of the output, it also increases the number of queries substantially.
%Next, we show that the bound in the above theorem is tight up to logarithmic factors.
%Specifically, we prove that any deterministic algorithm with $q$ adaptive halfplane range-emptiness queries must have error $\Omega(\frac{1}{q^2})$ in the worst case.
In fact, the bound $\widetilde{O}(1/q^2)$ is already tight up to logarithmic factors.
%The proof is deferred to the appendix.
We prove an even stronger statement: any deterministic algorithm with $q$ adaptive queries to a $0$-accurate extreme halfplane oracle must have error $\Omega(\frac{1}{q^2})$.

For a set $P$ of points in $[0,1]^2$, denote by $\mathcal{O}_P$ the $0$-accurate extreme halfplane oracle on $P$ and denote by $\mathcal{O}_P(\vec{v})$ the halfplane returned by $\mathcal{O}_P$ when queried by $\vec{v} \in \mathbb{S}^1$.
Consider an algorithm $\mathbf{A}$ that approximates $\mathcal{CH}(P)$ for an unknown set $P$ of points in $[0,1]^2$ via applying $q$ queries to $\mathcal{O}_P$.
Let $D$ be the maximum disk inside $[0,1]^2$ and $Z$ be a set of $4q$ points uniformly distributed on $\partial D$.
%Define $A = \{a_1,\dots,a_{4q}\}$ and $Z = \mathcal{CH}(A)$, which is a regular $4q$-gon.
%and $Z' = \mathcal{CH}(\{a_1,a_3,\dots,a_{4q-1}\})$.
%Then $Z$ is a regular $4q$-gon and $Z'$ is a regular $2q$-gon.
%Suppose $a_1,\dots,a_{4q} \in \partial D$ are the vertices of $Z$, sorted in clockwise order along $\partial D$.
%Define $Z' = \mathcal{CH}(\{a_1,a_3,\dots,a_{4q-1}\})$, which is a regular $2q$-gon.
Suppose $\vec{v}_1,\dots,\vec{v}_q \in \mathbb{S}^1$ are the $q$ queries applied by $\mathbf{A}$ to the oracle $\mathcal{O}_Z$ when running on $P = Z$, and let $H_i = \mathcal{O}_Z(\vec{v}_i)$ for $i \in [q]$.
We say a point $b \in Z$ is \textit{bad} if $b \in \partial H_i$ for some $i \in [q]$.
Let $B \subseteq Z$ be the set of bad points.
Observe that $|B| \leq 2q$.
Indeed, no three points in $Z$ are collinear and thus each $\partial H_i$ intersects at most two points in $A$, which implies that the number of bad points is at most $2q$.
%Now define $Z' = \mathcal{CH}(B)$.
%We have the following observation.
We first show that $\mathbf{A}$ makes the same queries when run on $\mathcal{O}_B$.

\begin{lemma}
    $\mathcal{O}_B(\vec{v}_i) = H_i$ for all $i \in [q]$.
\end{lemma}
\begin{proof}
Since $B \subseteq Z$, we have $\mathcal{O}_B(\vec{v}_i) \subseteq \mathcal{O}_Z(\vec{v}_i) = H_i$.
As $\mathcal{O}_Z$ is $0$-accurate, $H_i$ is the minimal halfplane with normal vector $\vec{v}_i$ such that $Z \subseteq H_i$.
Thus, $Z \cap \partial H_i \neq \emptyset$, which implies $B \cap \partial H_i \neq \emptyset$.
So $\mathcal{O}_B(\vec{v}_i)$ contains a point on $\partial H_i$, which implies $H_i \subseteq \mathcal{O}_B(\vec{v}_i)$ because $\mathcal{O}_B(\vec{v}_i)$ and $H_i$ are parallel.
\end{proof}

The above lemma implies that when running on $P = B$, the algorithm $\mathbf{A}$ performs exactly the same as running on $P = Z$.
As such, $\mathbf{A}$ cannot distinguish $Z$ and $B$, so the convex polygons output by $\mathbf{A}$ for $P = Z$ and $P = B$ are the same.
%The following observation shows that $C$ differs a lot from either $\mathcal{CH}(Z)$ or $\mathcal{CH}(B)$.

\begin{lemma}
    $\lVert \mathcal{CH}(Z) \backslash \mathcal{CH}(B) \rVert = \Omega(1/q^2)$.
\end{lemma}
\begin{proof}
Suppose $Z = \{a_1,\dots,a_{4q}\}$ where $a_1,\dots,a_{4q}$ are sorted along $\partial D$.
For $i \in [4q]$, write $\triangle_i = \triangle a_{i-1} a_i a_{i+1}$.
All $\triangle_i$'s have the same area, which we denote by $\eta$.
%We have $\triangle_i \subseteq Z$ for all $i \in [4q]$.
Observe that for each $i \in [4q]$, $\triangle_i$ is disjoint from the interior of $\mathcal{CH}(Z \backslash \{a_i\})$.
Thus, for every point $a_i \in Z \backslash B$, the interior of $\triangle_i$ is disjoint from the interior of $\mathcal{CH}(B)$.
Since $|B| \leq 2q$, we have $|Z \backslash B| \geq 2q$.
So there exists a subset $Z' \subseteq Z \backslash B$ with $|Z'| = q$ such that no two points in $Z'$ are adjacent in the circular sequence $(a_1,\dots,a_{4q})$.
Suppose $Z' = \{a_{i_1},\dots,a_{i_q}\}$.
Then the triangles $\triangle_{i_1},\dots,\triangle_{i_q}$ are pairwise interior-disjoint and all of them are disjoint from the interior of $\mathcal{CH}(B)$.
Since $\triangle_{i_1},\dots,\triangle_{i_q} \subseteq Z$, we have $\lVert \mathcal{CH}(Z) \backslash \mathcal{CH}(B) \rVert \geq \sum_{j=1}^q \lVert \triangle_{i_j} \rVert = q\eta$.
It suffices to show $\eta = \Omega(\frac{1}{q^3})$.
Consider a triangle $\triangle_i$.
The longest edge $\sigma$ of $\triangle_i$ is the segment connecting $a_{i-1}$ and $a_{i+1}$, whose length is $2\sin \frac{\pi}{2q}$.
The height of $\triangle_i$ with respect to $\sigma$ is $1-\cos \frac{\pi}{2q}$.
Thus, $\eta = \sin \frac{\pi}{2q} \cdot (1-\cos \frac{\pi}{2q})$.
By Taylor series of $\sin$ and $\cos$, we have $\sin \frac{\pi}{2q} = \frac{\pi}{2q} + o(\frac{1}{q})$ and $\cos \frac{\pi}{2q} = 1- (\frac{2\pi}{q})^2/2 + o(\frac{1}{q^2})$.
It follows that $\eta = \Omega(\frac{1}{q^3})$ and $\lVert \mathcal{CH}(Z) \backslash \mathcal{CH}(B) \rVert = \Omega(\frac{1}{q^2})$.
\end{proof}

Let $C$ be the output of $\mathbf{A}$ when running on $P = Z$, which is also the output when running on $P = B$.
We have $\lVert \mathcal{CH}(Z) \backslash \mathcal{CH}(B) \rVert \leq \lVert C \oplus \mathcal{CH}(Z) \rVert + \lVert C \oplus \mathcal{CH}(B) \rVert$.
Therefore, the above lemma implies that $\lVert C \oplus \mathcal{CH}(Z) \rVert = \Omega(\frac{1}{q^2})$ or $\lVert C \oplus \mathcal{CH}(B) \rVert = \Omega(\frac{1}{q^2})$.
It follows that with $q$ queries to a $0$-accurate extreme halfplane oracle, the minimum error one can achieve is $\Omega(\frac{1}{q^2})$.
As a $0$-accurate extreme halfplane oracle is stronger than a halfplane range-emptiness oracle, we have the following conclusion.

\begin{theorem} \label{thm-adapt_hplane_lb}
    Any deterministic algorithm for \textnormal{\sc Approximate Convex Hull} in $\mathbb{R}^2$ that performs $q$ adaptive halfplane emptiness queries has error $\Omega(1/q^2)$.
\end{theorem}

%\section{Conclusion}
%In this paper, we considered the \textsc{Approximate Convex Hull} problem:
%Given oracle access to an unknown point set $P$, output a convex body $C$ that closely resembles the convex hull $\mathcal{CH}(P)$.
%We presented tight and near-tight bounds on the tradeoffs between the number of queries $q$ made by the algorithm and the worst case error of the approximation, for both orthogonal range emptiness and halfplane emptiness queries in the adaptive and non-adaptive settings.
%A handful of questions remain.
%It would be interesting to see if one could resolve the $\log^{O(1)}q$ gap between the upper and lower bound for adaptive halfplane queries.
%One could also ask about the tradeoffs for other types of range emptiness queries.

\bibliographystyle{plainurl}
\bibliography{my_bib}

%\appendix
%\section{Missing proofs}
%\input{appendix}

\end{document}